\documentclass[11pt]{article}
\usepackage{fullpage}

\usepackage{float}
\usepackage{tikz}
\usepackage{amsthm}
\usepackage{amsmath}
\usepackage{amssymb}
\usepackage{graphicx}
\usepackage{algorithm,algorithmic}
\usepackage{esint,bbm,enumerate,color}
\usepackage{caption}

\newcommand{\sd}{\Sigma\Delta}

\newcommand{\R}{\mathbb{R}}

\newcommand{\A}{\mathcal{A}}
\newcommand{\supp}{\operatorname{supp}} 
 
\makeatletter

\theoremstyle{plain}
\newtheorem{thm}{\protect\theoremname}
\newtheorem{prop}[thm]{\protect\propname}
\theoremstyle{definition}
\newtheorem{defn}[thm]{\protect\definitionname}

\newtheorem{remark}[thm]{Remark}
\theoremstyle{plain}
\theoremstyle{plain}
\newtheorem{cor}[thm]{\protect\corollaryname}

\newtheorem{lm}[thm]{Lemma}

\theoremstyle{definition}

\providecommand{\corollaryname}{Corollary}
\providecommand{\propname}{Proposition}
\providecommand{\definitionname}{Definition}
\providecommand{\theoremname}{Theorem}
\providecommand{\assumptionname}{Assumption}

\makeatother

\providecommand{\corollaryname}{Corollary}
\providecommand{\definitionname}{Definition}
\providecommand{\theoremname}{Theorem}

\newcommand{\Z}{\mathbb Z}

\newcommand{\E}{\mathcal E}

\newcommand{\SD}{\Sigma\Delta}

\providecommand{\fr}{\frac}

\providecommand{\wtl}{\widetilde}

\usepackage{constants}
\newconstantfamily{errorterms}{symbol=d}
\newconstantfamily{littlec}{symbol=c}

\usetikzlibrary{shapes.geometric,shapes.arrows,decorations.pathmorphing}
\usetikzlibrary{matrix,chains,scopes,positioning,arrows,fit}

\begin{document}

\title{From compressed sensing to compressed bit-streams: practical encoders, tractable decoders}

\author{Rayan~Saab, Rongrong~Wang, and {\"O}zg{\"u}r~Y{\i}lmaz
  \thanks{Rayan~Saab is with the Mathematics Department of University
    of California, San Diego. Rongrong Wang and {\"O}zg{\"u}r~Y{\i}lmaz are with the
    Department of Mathematics, University of British Columbia,
    Vancouver, BC, Canada.}%
  }

\markboth{IEEE Transactions on Information Theory}{Your Name \MakeLowercase{\emph{et al.}}: Your Title}
\maketitle
\begin{abstract}
  Compressed sensing is now established as an effective method for
  dimension reduction when the underlying signals are sparse or
  compressible with respect to some suitable basis or frame. One
  important, yet under-addressed problem regarding the compressive
  acquisition of analog signals is how to perform quantization. This
  is directly related to the important issues of how ``compressed''
  compressed sensing is (in terms of the total number of bits one ends
  up using after acquiring the signal) and ultimately whether
  compressed sensing can be used to obtain compressed representations
  of suitable signals.  Building on our recent work, we propose a
  concrete and practicable method for performing
  ``analog-to-information conversion'' . Following a compressive
  signal acquisition stage, the proposed method consists of a
  quantization stage, based on $\sd$ (sigma-delta) quantization, and a
  subsequent encoding (compression) stage that fits within the
  framework of compressed sensing seamlessly. We prove that, using
  this method, we can convert analog compressive samples to compressed
  digital bitstreams and decode using tractable algorithms based on
  convex optimization. We prove that the proposed AIC provides a
  nearly optimal encoding of sparse and compressible signals. Finally,
  we present numerical experiments illustrating the effectiveness of
  the proposed analog-to-information converter.

\end{abstract}

\section{Introduction}\label{sec:intro}

An analog-to-information converter (AIC) collects compressive samples
of inherently analog signals and recovers these signals
using tractable algorithms (see, e.g., \cite{LaskaA2I, LaskaA2I_2}). We address the following outstanding
question in the compressive sampling literature: Can we design an AIC
such that (a) the sampling is compressive, (b) it results in a
(nearly) optimal encoding (in the sense of Kolmogorov) of the original
signals, (c) it is practicable?  The compressive AIC model that we
 focus on consists of a \emph{compressive sampling} stage, a
\emph{quantization} stage, an \emph{encoding stage}, and a
\emph{reconstruction} or \emph{decoding} stage where the signal of
interest is approximated, see Figure~\ref{fig:block_diagram}.

\begin{figure}[t]\begin{center}
\sffamily
\begin{tikzpicture}[node distance=5cm,auto]
  \matrix (m) [matrix of nodes, 
    column sep=20mm,
    row sep=1cm,
    nodes={draw, 
      line width=1pt,
      anchor=center, 
      text centered,
      rounded corners,
      minimum width=1cm, minimum height=8mm,
      scale=0.5
    }, 
    right iso/.style={isosceles triangle,scale=0.4,sharp corners, anchor=center, xshift=-4mm},
    left iso/.style={right iso, rotate=180, xshift=-8mm},
    txt/.style={text width=1cm,anchor=center,scale=0.8},
    ellip/.style={ellipse,scale=0.4},
    empty/.style={draw=none}
    ]
  {

  compressive sampler 
  & 
quantizer
  & 
encoder
& 
decoder
  \\
  };  
\node [coordinate] (end) [right of=m-1-4, node distance=2cm]{};
\node [coordinate] (start) [left of=m-1-1, node distance=3cm]{};
  \tikzset{blue dotted/.style={draw=blue!50!white, line width=1pt,
                               dash pattern=on 1pt off 4pt on 6pt off 4pt,
                                inner ysep=10mm, inner xsep=3mm, rectangle, rounded corners}};

  \tikzset{red dotted/.style={draw=red!50!white, line width=1pt,
                               dash pattern=on 1pt off 4pt on 6pt off 4pt,
                                inner ysep=10mm, inner xsep=2mm, rectangle, rounded corners}};

  \node (first dotted box) [blue dotted, 
                            fit =  (m-1-1) (m-1-3)] {};
  \node (second dotted box) [red dotted, 
                            fit =  (m-1-4) (m-1-4)] {};

    \path[->] (start) edge node [pos = .0]{${x}\in\mathcal{X},$} (m-1-1) ;
    \path[->] (start) edge node [below,pos = .0]{$\mathcal{X}\subset\R^N$} (m-1-1) ;
    
    \path[->] (m-1-1) edge node [above]{$y=\Phi x$} (m-1-2);  
    \path[->] (m-1-1) edge node [below]{$\Phi(\mathcal{X})\subset \R^m$} (m-1-2);  

    \path[->] (m-1-2) edge node {$q\in \mathcal{C}_0$} (m-1-3); 
        \path[->] (m-1-2) edge node [below]{$|\mathcal{C}_0| < \infty$} (m-1-3); 
        
    \path[->] (m-1-3) edge node [pos=0.45]{$c\in \mathcal{C}$} (m-1-4);  
    \path[->] (m-1-3) edge node [below, pos=0.52]{$|\mathcal{C}| < |\mathcal{C}_0|$} (m-1-4);  
  
    \path[->] (m-1-4) edge node [pos = .95]{$\hat{x}\in\R^N$} (end) ;

  \node at (first dotted box.north) [above, inner sep=3mm] {\textbf{Acquisition}};
  \node at (second dotted box.north) [above, inner sep=3mm] {\textbf{Reconstruction}};

\end{tikzpicture}\end{center}
\caption{A block diagram depicting a compressive AIC for the
  acquisition and reconstruction of a signal $x\in \mathcal{X}\subset
  {\R^m}$. As part of the \emph{acquisition} stage of the AIC, a
  compressive sampler produces a vector of measurements $y=\Phi x \in
  \R^m$. The measurements are then quantized, i.e., replaced by a
  vector $q$ from a finite set $\mathcal{C}_0$. The encoder then
  replaces $q$ by an element $c$ of an even smaller finite set
  $\mathcal{C}$, known as the codebook. Usually the codebook satisfies
  $\log_2|\mathcal{C}| \ll \log_2|\mathcal{C}_0|$, as this reduces the
  number of bits needed to represent $c$ compared to $q$. Finally, the
  decoder produces an estimate $\hat{x}$ of the signal $x$, using only
  $c$ and knowledge of the maps associated with the three acquisition
  stages. The goal of an AIC is to produce, in a computationally
  tractable way, a good approximation of $x$ with a small codebook $\mathcal{C}$. \label{fig:block_diagram}}

\end{figure}
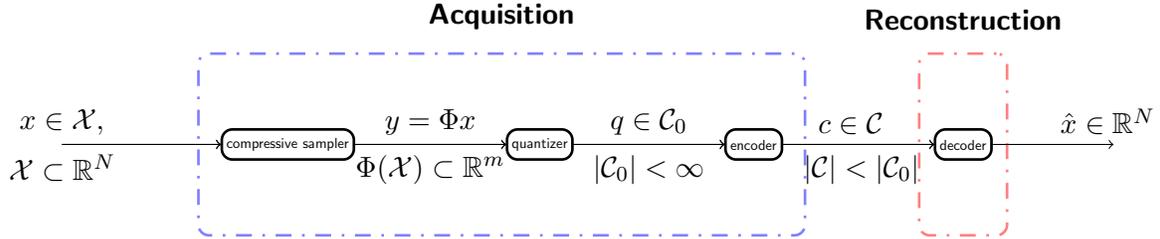

One of the main original insights of our approach is that the inclusion of
the encoding stage in this model makes it possible to answer the above
question affirmatively. To make the discussion concrete, we now discuss
the class of signal of interest and the individual stages of the
compressive AIC. 
\medskip





\noindent{\bf Signals:} We model signals as vectors in a fixed compact
set $\mathcal{X}$ in $\R^N$. Of particular interest are compressible
vectors, i.e., vectors
that can be well represented by their best $k$-term
approximation. These include bounded $k$-sparse signals as well as
signals in some fixed ball of a weak $\ell_p$ space  in
$\R^N$, denoted $w\ell_p$, with $0<p<1$ (see Section \ref{sec:def} for the precise definition).
\medskip

\noindent{\bf Compressive sampling:} A compressive sampling operator
$\Phi$ is an $m \times N$  matrix (typically with $m\ll N$) that provides
the vector of measurements $y=\Phi x$. The entries $y_i$ of the vector
$y$ are the compressive samples of $x$.
\medskip

\noindent{\bf Quantization:} The compressive samples must be
transmitted, stored, and processed using digital media. Therefore,
they need to be quantized: A quantization operator $\mathcal{Q}$ maps
$y\in \R^m$ to $q\in \mathcal{C}_0$ where $\mathcal{C}_0$ is a finite
set. Accordingly, the quantized measurements can be represented using
finite bitstreams. A notable special case, which we will mostly
restrict our attention to, is when $\mathcal{C}_0=\A^m$ for a finite set
$\A\subset \R$ called the quantization alphabet (an example is the ``1-bit" alphabet $\A=\{-1,1\}$).  Progressive quantizers such as
memoryless scalar quantization and $\sd$ quantization are of
this form.

\medskip

\noindent{\bf Encoding (Compression):} 
The quantized measurements require $\log_2|\mathcal{C}_0|$ bits to be
represented.  Often, we can reduce this bit budget by incorporating an
encoding stage. We denote by $\mathcal{E}: \mathcal{C}_0 \to
\mathcal{C},$ the encoding map, where $\mathcal{C}$ is a finite set
called the codebook, usually satisfying $\log_2|\mathcal{C}| \ll
\log_2|\mathcal{C}_0|$. The goal of encoding is to reduce the number of
bits while still permitting accurate reconstruction. We focus
  on simple encoding schemes implemented via, e.g., discrete
  Johnson-Lindenstrauss embeddings \cite{JLoriginal}
  (cf. \cite{ailon2006approximate, ailon2009fast, krahmer2011new,
    vybiral2011variant}).
%
%
\vspace{5pt}\\{\bf Reconstruction:} The final stage of a compressive
AIC is the reconstruction or decoding stage where we recover an
approximation to the original signal $x$. To that end we use a map $\Delta:
\mathcal{C} \to \R^N$.  Since we do not impose a probabilistic model on the signals,
it is natural to study the worst case reconstruction error, i.e.,
the distortion $\mathcal{D}$, in terms of the bit rate $\mathcal{R}$
where
 \begin{align} 
\mathcal{D}&:=
\sup\limits_{x\in\mathcal{X}}\|\Delta(\mathcal{E}(\mathcal{Q}(\Phi x))) -x\|_2\label{eq:worst_case},\\
\mathcal{R}&:= \log_2|\mathcal{C}|.\label{eq:rate}
 \end{align}
%
%
\subsection{Definitions and notation} \label{sec:def}
Throughout, for an $m\times N$ matrix $A$ and $T\subseteq \{1,\dots, N\}$ we denote by $A_T$ the submatrix formed by the columns of $A$ indexed by $T$. Similarly, for $x\in \R^N$, $x_T$ denotes the restriction of $x$ to $T$. We denote the set of $k$-sparse vectors in $\R^N$ by 
$$\Sigma_k^N :=\{ x\in \R^N, |\supp(x)|\leq k  \}.$$

\begin{defn}
We say that a vector $x \in \R^N,$ belongs to the weak $\ell_p$ ball
of radius $C$ if $|x|_{(j)} \leq Cj^{-1/p}$ where $|x|_{(j)}$ denotes
the magnitude of the $j$th largest-in-magnitude entry of $x$. 
\end{defn}

\begin{defn}[The restricted isometry property]
  We say that an $m\times N$ matrix $A$ satisfies the Restricted
  Isometry Property (RIP) of order $k$ and constant $\delta_k$ if for
  all $k$-sparse vectors $x$ we have
\begin{equation}
(1-\delta_k) \|x\|_2^2 \leq \|A x\|_2^2 \leq (1+\delta_k) \|x\|_2^2. 
\end{equation}

\end{defn}


\begin{defn}[Sub-Gaussian random variables and matrices]
\hfill

\begin{enumerate}[(i)]
\item A random variable $\eta$ is sub-Gaussian with parameter $c>0$ if
  it satisfies  $P(|\eta|>t) \leq e P(|\xi|>t)$ where $\xi$ is a
  Gaussian random variable with mean 0 and variance $c^2$. 
\item A matrix $E$ is sub-Gaussian with parameter $c$, mean $\mu$ and
  variance $\sigma^2$ if its entries are independent sub-Gaussian
  random variables  with parameter $c$ , mean $\mu$, and variance $\sigma^2$.
\end{enumerate}
\end{defn}

We remark that one can also define sub-Gaussian random variables via
their moments or, when they are zero mean, via their moment generating
functions. See \cite{ve12-1} for a proof that all these definitions
are equivalent.
Note that Gaussian random variables and all bounded random variables
(e.g., Bernoulli), as well as their linear combinations are
sub-Gaussian random variables. 

Among all sub-Gaussian random matrices,
we give special emphasis to \emph{Bernoulli matrices} (because
we use them for encoding our quantized
measurements): 
A matrix $B$ is said to be a Bernoulli matrix if each of its entries
$B_{ij}$ is drawn randomly from $\{\pm1\}$ such that
$\mathbb{P}(B_{ij}=
1)=1/2$. 

Throughout, we write $a(x)\lesssim b(x)$ if and only if there exists a constant $C$ such that $a(x)\le C b(x)$. 
Given a set $\mathcal{X}\subset \R^N$, we denote its image under a map $f$ by  $f(\mathcal{X}):=\{f(x), x \in \mathcal{X}\}$. Finally, note that we use the terms ``compressed sensing" and ``compressive sampling" interchangeably.

\subsection{Main contributions} Below, the compressive
 sampling matrix $\Phi$ is an $m\times N$ sub-Gaussian matrix. The
 measurements $y=\Phi x + e$ are possibly corrupted by noise $e$ with
 $\|e\|_\infty \le \epsilon$. The quantization operator $\mathcal{Q}$
 is an $r$th-order Sigma-Delta ($\sd$) scheme yielding $q=\mathcal{Q}(y)=y-D^r
 u$; here $u$ is a bounded state vector and $D$ is the bidiagonal
 matrix with entries on the main diagonal equal to 1 and on the
 subdiagonal equal to $-1$ (see Section~\ref{sec:SD}).  The action
 of the encoding map $\mathcal{E}$ can be decomposed into two
 stages. First, we apply an $L\times m$ Bernoulli matrix $B$ to
 $D^{-r}q$ so that by construction $BD^{-r}q$ takes values from a finite
 set. The second stage of the encoding is simply to assign binary
 labels to $BD^{-r}q$, which is an invertible operation and will be mostly ignored
 in our statements below. The reconstruction operator $\Delta$ is a modified version of the one we proposed in \cite{SWY15}, where no encoding was assumed.
It is based on solving the convex optimization problem
\begin{equation}\label{eq:opt_int}
(\hat{x},\hat{u}, \hat{e})=\arg\min \|\widetilde{x}\|_1, \text{  s.t. } \left\{\begin{array}{rcl} 
BD^{-r}(\Phi \widetilde{x}+\widetilde{e}) -B\widetilde{u} &=&BD^{-r}{q}\\
\|B\widetilde{u}\|_2 &\leq& 3Cm\\
\|\widetilde{e}\|_2 &\leq & \sqrt{m} \varepsilon
\end{array}.\right.
\end{equation}
Here, $C$ is a known constant that depends on the specific $\SD$
quantizer used.  Our main theorem is as follows.

\begin{thm}\label{thm:main_intro} 
With high probability, the following holds for all $x$ that
satisfy  $\|\Phi x\|_\infty \le \mu < 1$ where $\mu$ is a fixed
constant. 
Let $q:=\mathcal{Q}(\Phi x+e)$,  where  $\|e\|_\infty\le \varepsilon$ for some $0\leq	\varepsilon<1-\mu$.  Then the solution $\hat{x}$ to \eqref{eq:opt_int} satisfies
\begin{equation}\label{eq:err01}
\|x-\hat{x}\|_2 \lesssim \Big(\frac{L}{m}\Big)^{r/2-3/4} + \sqrt{\frac{m}{L}}\varepsilon + \frac{\sigma_k(x)}{\sqrt{k}},
\end{equation}
for all $k$ satifying $L\ge C_0 k \log N$ where $C_0$ is a constant
that depends on $\Phi$ and $r$.  

\end{thm}

\begin{remark}Examining \eqref{eq:err01}, our results are meaningful for quantizers of order $r\ge 2$
  as that ensures the exponent $r/2-3/4$ in \eqref{eq:err01} is positive. When
  $r=1$, we have an alternative approach that yields analogous results
  with an improved exponent $r/2-1/4$, though only in the strictly
  sparse and noiseless case. See Section~\ref{sec:main_results}, Theorem \ref{thm:impv}. 
\end{remark}

The error bound in Theorem \ref{thm:main_intro} is comparable to the analogous bound obtained in \cite{SWY15}, which studied $r$-th order $\sd$ quantization \emph{without} encoding. In particular, the error bound associated with $r$-th order $\sd$ quantization from \cite{SWY15} differs from \eqref{eq:err01} only in the exponent associated with the $\frac{L}{m}$ term: $r-1/2$ in \cite{SWY15} versus $r/2 -3/4$ in our case. While \cite{SWY15} still obtained root-exponential error decay (in the bit-rate) by selecting an optimal order for the quantizer, here we obtain \emph{exponential} error decay \emph{for every} order $r\geq2$ by incorporating encoding, as we see in  Corollary \ref{corr:nearop} below. 
We omit technical details here for the sake of clarity, see Section~\ref{sec:main_results} for the full versions. 

\begin{cor}[\bf Exponential error decay and near-optimal encoding of sparse signals] \label{corr:nearop} In the noise-free case, i.e., when $\epsilon=0$, we have  
\begin{equation}\label{expdecaycomp}
\mathcal{D} \lesssim
2^{-\Cl{c_delta} \frac{\mathcal{R}}{k_0\log N}} +
\frac{\sigma_{k_0}(x)}{\sqrt{k_0}}
\end{equation}
where $k_0:=\lfloor \frac{L}{C_0 \log N}\rfloor$. Furthermore, if $x$
is $k$-sparse with $k\le k_0$, this shows that
$$\mathcal{D}\lesssim 2^{-\Cl{c_delta} \frac{\mathcal{R}}{k_0\log
    N}},$$
i.e., we have exponential accuracy. 
\end{cor}

\begin{remark}[\bf One-bit quantization] Earlier use of $\SD$ quantization in the compressed sensing setup \cite{GLPSY13, KSY13} was restricted to multi-bit alphabets. This was primarily due to the decoder proposed in \cite{GLPSY13} incorporating a support-recovery stage. As our decoder is now solely based on solving the optimization problem \eqref{eq:opt_int}, it allows one-bit compressed sensing albeit with a $\SD$ quantizer. The advantage is that Theorem \ref{thm:main_intro} and Corollary \ref{corr:nearop} apply and we therefore have \emph{exponential error decay with a one-bit quantizer.} 
\end{remark}

 \begin{remark}[\bf Near-optimal \emph{compressive} encoding] As seen in \eqref{expdecaycomp},
    in the no-noise setting the error decays exponentially fast in the bit rate until it hits the
   best $k_0$-term approximation error. Noting that the AIC scheme we propose effectively reduces
   the number of measurements to $L$, consider the same
   scheme but replace the quantization with the identity map (i.e.,
   do not quantize). We now have a compressive sampling scheme with an
   $L\times N$ measurement matrix. Accordingly, the classical CS results (e.g., \cite{CRT05}) can at
   best guarantee a reconstruction error of
   $\sigma_{k_0}(x)/\sqrt{k_0}$ where $k_0$ is as in
   Corollary~\ref{corr:nearop}. Our scheme approaches this optimal
   lower bound \cite{cohen2009compressed} exponentially fast in $\mathcal{R}$. 
\end{remark}

\begin{remark}[\bf Near-optimal encoding for the $w\ell_p$ ball] The rate-distortion relationship when encoding the unit $w\ell_p$ ball $B_{w\ell_p}$ (using any method) satisfies 
\begin{equation}\mathcal{D}(B_{w\ell_p}) \gtrsim \left( \frac{1}{\mathcal{R}} \log(N/\mathcal{R}+1) \right)^{1/p-1/2} \quad \text{whenever}  \quad \log_2 N \leq \mathcal{R}\leq N.\label{eq:wlp_lower}\end{equation}
In fact, this bound is optimal for the unit $\ell_p$ ball (cf. \cite{edmunds1996function,kuhn2001lower, schutt1984entropy}), which is slightly smaller than $B_{w\ell_p}$. Note that the range of $\mathcal{R}$ given above is appropriate for our setting. When $\mathcal{R}$ is smaller, one cannot obtain meaningful bounds. On the other hand, when $\mathcal{R} > N$, the optimal error bound (which is attained by directly encoding the coefficients of $x$, say using MSQ) decays exponentially in $\mathcal{R}$.  With our AIC, using the observation $\frac{\sigma_k(x)}{\sqrt{k}} \lesssim k^{1/2-1/p}$ when $x$ is in $B_{w\ell_p}$, we can show that an appropriate choice of $L$ and $m$ yields the rate-distortion relationship
\begin{equation}\label{eq:wlp_upper}
\mathcal{D}(B_{w\ell_p}) \lesssim  \left(\frac{\log N \log(\mathcal{R}/\log N- \log \log N)}{\mathcal{R}-\log N \log\log N}\right)^{1/p-1/2}.
 \end{equation}
Whenever $\mathcal{R}\gtrsim \log N \log\log N$, this yields
\[\mathcal{D}(B_{w\ell_p}) \lesssim \left(\frac{\log N\log (\mathcal{R}/\log N )}{\mathcal{R}}\right)^{1/p-1/2}.\]
Except for a $\log (\mathcal{R}/\log N)$ factor, this is optimal.
\end{remark}

\begin{remark}[\bf Robustness] $\SD$ quantizers are robust to certain circuit imperfections which makes them popular in practical applications (see, e.g., \cite{daub-dev, NST96}). The proposed AIC inherits all such favorable properties of $\SD$ schemes. 
\end{remark}

\begin{remark}[{\bf Bounded measurements}] The condition $\|\Phi x\|_\infty \leq \mu <  1$ (from Theorem \ref{thm:main_intro}) is a natural one in any quantization context, as it ensures that finite alphabets can be used for quantization. Moreover, there are many regimes where such a condition is easily satisfied.  For example, when the
  entries of $\Phi$ are bounded random variables, 
  $\|\Phi x\|_\infty < \mu$ can be guaranteed independently of $m$ and
  $k$.  This is the case when $\Phi$ is a Bernoulli matrix with
  i.i.d. $\pm1$ entries, and $x\in
  \mathcal{B}:=\mu B_{\ell^N_1}$ where $B_{\ell^N_1}$ denotes the unit
  $\ell_1$ ball in $\R^N$. Similarly, when the entries of $\Phi$ are sub-Gaussian random variables and $x$ is drawn from the $w\ell_p$ ball, the boundedness of $\Phi x$ can be guaranteed with high probability (and with probability $1$ if $\Phi$ is, say, Bernoulli). 
 \end{remark}

\subsection{Roadmap}
In Section \ref{sec:setup}, we discuss the motivation for studying quantization and encoding problems in the compressed sensing setup. Section \ref{sec:prior} discusses the prior work in this area,  while Section \ref{sec:encoding} describes the proposed AIC in detail. Section \ref{sec:main_results} provides the technical statements of our main results, the proofs of which are provided in Sections \ref{sec:proof1} and \ref{sec:proof2}.


\section{The quantization and encoding problems in compressed sensing}\label{sec:setup}
Optimal encoding of signals in a given compact metric space $(\mathcal{X},d)$ is
an approximation theoretic problem. Given an acceptable approximation
error margin $\epsilon$, one seeks to cover $\mathcal{X}$ with the smallest
number of $\epsilon$-balls associated with the metric $d$. This
number, denoted by $\mathcal{N}(\mathcal{X},d,\epsilon)$, is called the covering
number of $\mathcal{X}$. In essence one can then \emph{encode} $x\in \mathcal{X}$ using
$\log_2 \mathcal{N}(\mathcal{X},d,\epsilon)$ bits (this quantity is called the
Kolmogorov $\epsilon$-entropy of $\mathcal{X}$) by mapping it to the center of
an $\epsilon$-ball in which $x$ lies. The set $\mathcal{C}$ of all
such centers (or ``codewords'') is called the codebook of the encoder. Clearly, the resulting
approximation error (or ``distortion'') is a decreasing function of
the number of bits required to encode a signal (``rate''), thus there
is a rate-distortion trade-off. For example, for the set of bounded $k$-sparse signals $\mathcal{X}=\Sigma_k^*:=\{x\in\R^N, \|x\|_2 \leq 1, |\supp{x}| \leq k\}$ and the $\ell_2$ metric, the optimal rate distortion relationship \cite{BB_DCC07} (see also \cite{BJKS15}) is: 
\begin{equation}\mathcal{D}_{\Sigma_{k}^{*}}(\mathcal{R}) \gtrsim \frac{N}{k}2^{-\frac{\mathcal{R}}{k}}.\label{eq:sparse_lower_bound}\end{equation} This entropy based approach to encoding, while useful for
 providing the optimal rate-distortion relationship of a given
signal class, is not practicable in the compressed sensing scenario (among others) for various reasons:

\begin{itemize}
\item \emph{It assumes direct access to $x$.} This rules out the
  compressed sensing setup because recovering $x$ from \emph{unquantized}
  measurements requires implementing a compressed sensing decoder on
  analog hardware. As such decoders involve solving large scale
  convex optimization problems or implementing greedy algorithms, they
  are not amenable to analog computation. 
\item \emph{It suffers from the curse of dimensionality.} Even if we
  ignore the above issue, as $\epsilon$ decreases,
  $\mathcal{N}(X,d,\epsilon)$ typically grows fast. For example,
  $\mathcal{N}(B_2^N,d_2,\epsilon)$ scales like $\epsilon^{-N}$ where
  $B_2^N$ denotes the unit $\ell_2$-ball of $\R^N$ and $d_2$ is the
  metric induced by the $\ell_2$ norm. Accordingly, to achieve a
  distortion of $\epsilon$ one must enumerate and accurately store the
  $O(\epsilon^{-N})$ points in the codebook and subsequently compute
  the distance of arbitrary points to them. This can quickly become
  prohibitive as $\epsilon$ decreases.
\item \emph{It is not robust with respect to hardware imperfections.}
  To correctly assign a signal $x$ to a codeword,  the hardware must
  distinguish analog values that are separated by $O(\epsilon)$. This
  is expensive and quickly becomes prohibitive as $\epsilon$
  decreases. Furthermore, even a small error in the comparison may
  lead to selecting the ``wrong'' codeword, and hence an error exceeding
  $\epsilon$.   
\end{itemize}

A practical quantization scheme in the CS setting must
avoid the issues listed above. Moreover, depending on the 
implementation details, it may be desirable (or even crucial in some cases) that the quantizer possess the following
properties. 

\begin{enumerate}[(P1)]
\item \emph{It should be compatible with the requirements of (say, state-of-the-art) analog to digital conversion.} For example, the scheme should not store more than a few analog quantities, or store them for too long, or require sophisticated analog computation. 
\item \emph{It should admit a computationally efficient reconstruction algorithm, or decoder.}
\item \emph{It should be universal.} The quantizer should not use any prior
  information about the measurement scheme or the signal. In
  particular, the quantization scheme should be a ``black
  box'' that can be placed after any CS measurement system. 
\item \emph{It should be causal.} Many important CS systems such as the
  single-pixel camera \cite{duarte2008single} and those based on coded-aperture imaging (see, e.g., \cite{willett2011compressed, wagadarikar2008single}) obtain the measurements
sequentially. In such instances,  a quantizer should not assume
knowledge about ``future'' measurements.     
\item \emph{It should be progressive.} It should be able to
  incorporate any additional measurements to improve the approximation
  accuracy.
\item \emph{It should be coarse.}  Given a fixed quantization
  alphabet, it should allow an arbitrarily accurate approximation to
  the original signal by increasing the number of
  measurements. 
\item \emph{It should be robust.} Any quantization scheme must involve
  certain arithmetic and Boolean operations, e.g., addition and
  comparison with a reference value. These operations cannot be
  implemented with infinite accuracy on analog circuits because of
  physical limitations (see, e.g., \cite{NST96}).  Thus, a practical quantization scheme must be
  robust with respect to such imperfections.
 \end{enumerate}
Given a practical quantizer, often one can incorporate an encoding stage and reduce the total number of bits used to represent the signal. With such an encoding stage, the rate-distortion trade-off can be observed by considering the distortion $\mathcal{D}$ (as in \eqref{eq:worst_case}) as a function of the final bit-rate after encoding, i.e., $\mathcal{R}$ as in \eqref{eq:rate}. 
This is the approach we follow in this paper and it stands in contrast to the case where no encoding is present (as in, e.g., \cite{GLPSY13,SWY15,baraniuk2014exponential}) and the final bit-rate is a constant multiple of the number of measurements.


\section{Relevant prior work}\label{sec:prior}
There has been growing interest in quantization for compressed
sensing, which has resulted in a number of important
contributions. Broadly speaking, the techniques proposed in the
literature fall in one of the two main quantization paradigms: fine
quantization or coarse quantization. In \emph{fine quantization}, one
achieves the desired accuracy by refining the finite quantization
alphabet $\A$, or, equivalently, reducing the quantization step size
$\delta$ (defined as the largest gap between two consecutive elements
of $\A\subset\R$). In this case, it is easy to obtain exponential
accuracy in terms of the bit budget as one can use $n$ additional bits
to reduce the stepsize $\delta$ by a factor of $2^{n}$. One can then
use any quantization method that ensures $\|y-q\|_\infty \lesssim
\delta$; consequently, any robust CS recovery algorithm will yield an
approximation with accuracy on the order of $\delta$.  Noting that
these small step sizes need to be accurately implemented on analog
hardware, a major shortcoming of fine quantization algorithms lies in
the difficulty in (and sometimes the impossibility of) reducing
$\delta$ to sufficiently small values due to physical constraints.
Therefore, quantizers with small step size $\delta$ are more expensive
and there is a physical lower bound on how small $\delta$ can be. On
the other hand, in \emph{coarse quantization} one uses a fixed
alphabet, possibly as coarse as 1-bit, and improves accuracy by
increasing the number of measurements. Accordingly, coarse quantizers
are typically cheap to implement robustly on analog hardware. However,
obtaining nearly optimal rate distortion characteristics, i.e.,
exponential decay of approximation error as a function of the bit
budget, is highly non-trivial (see, e.g., \cite{CGKSY15}).  We now
provide a brief (non-exhaustive) overview of the literature that is
most related to our work. We focus primarily on memoryless scalar
quantization and on $\sd$ quantization but we also give some attention
to the special case of one-bit quantization due to the attention it
has recently received.  More detailed reviews can be found in
\cite{boufounos2014quantization,CGKSY15}.

We begin by observing that many of the quantizers studied in the literature
fall under the umbrella of memoryless scalar quantization, while only some
use a noise-shaping (see, e.g., \cite{CGKSY15}) approach, of which $\sd$ quantization is an example.  All these
quantizers can be used within the fine quantization paradigm as one
can make the quantizer step size as small as the desired
reconstruction accuracy demands require. On the other hand, noise
shaping quantizers are much better suited to the coarse quantization
approach as sufficiently increasing the number of measurements can
meet any reconstruction accuracy demand without a need to change the
step size-- see below for more details.

\subsection{Memoryless scalar quantization}
Memoryless scalar quantization (MSQ) is possibly the simplest (but
certainly not the most efficient) way to quantize compressed sensing
measurements. Given an alphabet $\mathcal{A}\subset \R$, MSQ applies
\emph{scalar quantization} to each measurement independently by
replacing the measurement by the element of $\A$ nearest to it. More
precisely, define the scalar quantizer associated with $\mathcal{A}$,
$Q_\A: \R \to \A$ by
 \begin{align*}
 Q_\A(z)& \in\arg\min\limits_{v\in\A}|v-z| \label{eq:scalar_quantizer}
\end{align*} 
(if there are two minimizers $v_1,v_2$, pick $Q_\A(z)$ to be any one
of these) and denote by $y=\Phi x$ the compressive measurements. We
define the MSQ operator $Q_{\mathcal{A}}^{MSQ}: \Phi(\mathcal{X})
\rightarrow \mathcal{A}^m$ by
\begin{align*}
q=Q_{\mathcal{A}}^{MSQ}(y), \textrm{\quad where \quad} q_i=Q^{}_\mathcal{A}(y_i).
\end{align*}
Here, $\mathcal{X}\subset \R^N$ is the space of signals and $\Phi$ is the compressed sensing matrix. 

Usually the set $\A\subset \R$ is centered around zero and consists of elements $v$ separated by a quantization step size $\delta$.  In this case, provided the measurements $y_i$ are appropriately bounded, we have $|y_i-Q_\mathcal{\A}(y_i)| \leq \delta/2$. Thus, one approach to decoding MSQ-quantized CS measurements is to treat the quantization error $y-q$ as bounded measurement error (as considered in say, \cite{CRT05}) and to approximate $x$ from $q=Q_{\mathcal{A}}^{MSQ}(y)$ using the solution to the so-called Basis Pursuit De-Noising (BPDN) optimization problem as in \cite{Donoho2006_CS, CRT05}. This yields the decoder $\Delta_{\textrm{BPDN}}$, given by
\[
\Delta_{\textrm{BPDN}}(q):=\arg\min\limits_z \|z\|_1 \textrm{ subject to } \|\Phi z-q\|_2 \leq \delta\sqrt m/2.
\] 
Provided $\Phi$ is an appropriately chosen matrix (for example, satisfying the restricted isometry property \cite{CRT05})\footnote{The standard definition of the  restricted isometry property of a matrix $\Phi$ normally requires the columns of $\Phi$ to be normalized. Since we are interested in the quantization context where the number of measurements (hence the number of rows of $\Phi$) is variable, we do not normalize the columns of $\Phi$. 
},  the reconstruction error obeys
\begin{equation}\label{eq:noise}
\|\hat{x}-x\|_2 \lesssim \delta/2+ \frac{\sigma_k(x)}{\sqrt k}, 
\end{equation}
with $\hat{x}=\Delta_{\textrm{BPDN}}(q)$ and  $\sigma_k(x)=\min_{z\in \Sigma_k} \|x-z\|_1 $ being the $\ell_1$ error associated with the best $k$-term approximation of $x$. Note that, in the case of sparse signals, replacing $\delta$ by $2^{-n}\delta$ in \eqref{eq:noise} requires $n$ additional bits. Thus MSQ yields exponential accuracy \emph{when considered within the fine quantization paradigm}. On the other hand, 
once the quantization alphabet $\mathcal{A}$ is fixed (e.g., when the AIC hardware is fixed), the step-size $\delta$ is also fixed. In this case we are in the coarse quantization paradigm and the rate-distortion relationship associated with 
MSQ quantization and BPDN reconstruction of $k$-sparse signals is given by $$\mathcal{D}(\mathcal{R}) = \text{Constant}.$$
This is far from satisfactory as increasing the number of measurements, hence the rate, does not decrease the distortion. Moreover, the lower-bound for MSQ when $\delta$ is fixed (which is associated with optimally decoding MSQ quantized measurements) is only slightly better. It satisfies 
\begin{equation}
\mathcal{D}^{MSQ}(\mathcal{R}) \gtrsim \frac{k}{\mathcal{R}}.
\label{eq:lower_bound_MSQ}
\end{equation}
as derived using a frame-theoretic argument \cite{GVT98} (cf. \cite{boufounos2014quantization}). 
Consequently, even with a decoder that is optimal for MSQ-quantized CS measurements, one cannot hope to achieve the exponential rate-distortion relationship \eqref{eq:sparse_lower_bound} associated with entropy based encoding of sparse signals.  Nevertheless, there has been much work (e.g., \cite{zymnis2010compressed,QIHT,Jacques2010,Jacques2013}) focused around proposing  decoders to improve the reconstruction error associated with MSQ quantization of compressed sensing measurements and approach the lower-bound \eqref{eq:lower_bound_MSQ}.

\subsection{Sigma-Delta quantization for compressed sensing}\label{sec:SD}

\subsubsection*{Sigma-Delta quantization}\label{sec:SD} Let $y=\Phi x\in \R^m$ be as above. The simplest $\sd$ quantizer, known as the first order greedy $\sd$ scheme, maps $y$ to $q\in \A^m$ by running the iteration
 \begin{align}\label{eq:SD_greedy}
q_i &= Q_\A\left(y_i + u_{i-1}\right), \notag \\
(\Delta u)_i& := u_i-u_{i-1}=   y_i - q_i.
\end{align}
It simply consists of scalar quantizing the sum of the current measurement $y_i$ with a state variable $u_{i-1}$, and subsequently updating the state variable. More generally, a generic $r$th order $\sd$ quantizer maps $y$ to $q\in \A^m$ by running the iteration
\begin{align}
q_i &= Q_\A\left(\rho_r(u_{i-1},\dots,u_{i-r},y_i,\dots,y_{i-r+1})\right), \notag \\
(\Delta^r u)_i& =   y_i - q_i.
\label{equ:rthOrdu_gen}
\end{align}
Above the $r$th order difference operator $\Delta^r$ is defined via $\Delta^r(u):= \Delta(\Delta^{r-1} u)$, and the vector $u=(u_i)_{i=1}^m$ is called the state vector. It is typically ``initialized to zero'', i.e.,  $u_i=0$ for $i\le 0$. Similarly, the ``input'' $y$ is also initialized to zero. Note that in this case, the relationship between the vectors $u$, $y$, and $q$ can be described by the matrix equation 
\begin{equation}\label{eq:u}
y-q=D^r u
\end{equation}
where $D$ is the difference matrix defined in Section~\ref{sec:intro}. 

\subsubsection*{$\sd$ quantization and stability}The function $\rho_r$ in \eqref{equ:rthOrdu_gen} is called the quantization rule and is chosen to ensure that the $\sd$ scheme is \emph{stable}, i.e., there exist universal constants $\beta$ and $\gamma$ (independent of $y$ and $m$) such that $\|u\|_\infty \le \gamma$ whenever $\|y\|_\infty\le \beta$. Stability plays an important role in practice and also in the analysis of $\sd$ schemes due to the appearance of terms involving $u$ in error estimates. There are two main approaches for designing a quantization rule that ensures stability of an $r$th-order scheme. An $r$th-order \emph{greedy} $\sd$ quantizer  uses the quantization rule
\begin{equation}
\rho_r(u_{i-1},\dots,u_{i-r},y_i,\dots,y_{i-r+1}):=Q_\A\left(\sum\limits_{j=1}^r (-1)^{j-1} \binom{r}{j} u_{i-j}+y_i\right).
\end{equation}
where the alphabet $\A$ is tailored to the order $r$. A typical choice
for $\A$ is the \emph{$K$-level midrise alphabet with step size
  $\delta$} given by  
\begin{equation} \label{eq:alph}
\mathcal{A}^K_{\delta}:=\left\{\pm (j-1/2)\delta, 
j\in\{1,...,K\}\right\}.
\end{equation}
In this case, one chooses $K$ as a function of
$r$. Specifically, if $\|y\|_\infty \le \beta$, it is sufficient to
choose $$K\ge 2\big\lceil \frac{\beta}{\delta}\big\rceil+2^r+1$$ as such a
choice yields a stable $r$th-order $\sd$ scheme with stability
constant $\gamma=\delta/2$.  On the other hand, \emph{coarse $\sd$ quantizers} use a fixed alphabet $\A$ regardless of the order $r$, e.g., $\A=\{\pm 1\}$ or $\A=\A_\delta^K$ with $\delta$ and $K$ fixed. Designing families of stable $\sd$ schemes of arbitrary order is highly non-trivial, e.g., \cite{daub-dev,G-exp,DGK10}. We will use the schemes that were originally proposed in \cite{G-exp} and refined in \cite{DGK10} in the setting of 1-bit quantization. 
These $r$th-order coarse schemes use the alphabet $\A_\delta^K$ and produce state vectors $u$ that satisfy $$\|u\|_\infty \lesssim C^r r^r \delta$$ whenever $\|y\|_\infty \le (K-1/2)\delta$.


Originally proposed to quantize oversampled bandlimited functions by Inose and Yasuda \cite{inose1963unity}, cf. \cite{NST96}, $\sd$ quantization has been shown to be well suited for quantizing redundant frame expansions. This fact holds both in finite dimensions (e.g., \cite{benedetto2006sigma, BP07, BLPY10, KSW12, GLPSY, KSY13} and infinite dimensions \cite{daub-dev, G-exp, DGK10}. In each case, $\sd$ quantizers ``shape'' the quantization error such that a significant portion of the error energy falls into the kernel of the corresponding decoder. For example, in the case of oversampled bandlimited functions, the decoder can be described as the convolution of the quantized samples with an appropriate low-pass filter \cite{daub-dev}.  In the case of finite frames, on the other hand, a typical decoder is given by applying an alternative dual\footnote{We view a frame for $\R^d$ as a full-rank  $m\times d$ matrix $E$ where $m\geq d$ and we identify any left inverse of $E$ with a so-called \emph{dual frame} of $E$. In particular, the \emph{canonical} dual is identified with the Moore-Penrose pseudo-inverse $E^\dagger:=(E^*E)^{-1}E^*$, while the Sobolev dual is given by $(D^{-r}E)^\dagger D^{-r}$. } known as the Sobolev dual operator \cite{BLPY10}.

\subsubsection*{In the compressed sensing setting}It was recently shown by \cite{GLPSY}, cf. \cite{KSY13}, that $\sd$
quantization can be also used in the compressed sensing setting. The
idea is that if the signal $x$ is strictly sparse with support $T$,
$y=\Phi x= \Phi_Tx_T$ which is the vector of frame coefficients of
$x_T$. Here $\Phi_T$ is the (analysis operator of the) corresponding
frame. Thus, \cite{GLPSY} proposed a two stage method where in the
first stage one recovers the support $T$ using any robust compressed
sensing decoder. In the second stage, one uses the Sobolev dual of
$\Phi_T$ to obtain a finer estimate of $x$. This two-stage method is
successful in the case of strictly sparse signals whose smallest
entries are as large as the quantizer step size. In particular, when the entries of $\Phi$ are Gaussian \cite{GLPSY} or sub-Gaussian \cite{KSY13} random variables, the two-stage reconstruction method produces an estimate $\hat{x}$ that satisfies
\begin{equation}
\|\hat{x}-x\|_2 \lesssim \big(\frac{m}{k}\big)^{-({r}/{2}-{1}/{4})}\delta. \label{eq:err_GLPSY}
\end{equation}
As in the MSQ case, since the error is proportional to the step-size $\delta$, one obtains exponential accuracy in the bit-rate \emph{within the fine quantization paradigm}. On the other hand, when $\delta$ is fixed, i.e., in the coarse quantization paradigm, \eqref{eq:err_GLPSY} shows that the approximation error decays polynomially in the number of measurements, hence the rate. When $r\geq 2$ this is a faster decay rate than that of MSQ, which is limited by \eqref{eq:lower_bound_MSQ}.
 However,   \eqref{eq:err_GLPSY} is contingent on the success of the support
recovery stage, which becomes problematic in the case of sparse vectors with
non-zero entries that are much smaller in magnitude than the quantizer
step size. For example, this rules out quantizers with low bit-depth such as one-bit
quantizers. In addition, the two-stage method is not well suited for
compressible signals and for noisy measurements.

To overcome all these issues, in recent work we proposed a
one-stage decoder for $\sd$-quantized compressed sensing measurements
\cite{SWY15}. This one-stage decoder is based on solving a tractable
convex optimization problem and it allows us to remove the above
mentioned size condition, thus allowing quantization alphabets as
coarse as one-bit. Furthermore, this decoder is stable and robust,
i.e., it can be used with compressible signals in the presence of
noise. It also yields an approximation error bound that decays polynomially in
the number of measurements, and this again outperforms the optimal error decay
associated with MSQ---as given in \eqref{eq:lower_bound_MSQ}---for
$\sd$ schemes of order $r\ge 2$.

\subsection{One-bit quantization for compressed sensing}
A ``one-bit" quantization scheme is one where the alphabet $\mathcal{A}$ contains only two elements, with the usual choice being $\A=\{\pm 1\}$. There are multiple approaches to one-bit quantization for compressed sensing, including those based on MSQ, $\SD$, and other noise-shaping techniques. Among these, one-bit MSQ (e.g., \cite{BB_CISS08,jacques2013robust,PV13}) has received significant attention, usually under the monicker "one-bit compressed sensing".  One-bit MSQ schemes with $\mathcal{A}=\{\pm 1\}$ produce quantized measurements $q_i=\textrm{sign}(\langle \phi_i, x \rangle )$ and have the advantage of being simple to implement. On the other hand, due to the minimal size of the alphabet there are unique challenges associated with one-bit MSQ.  For example, since constant multiples of $x$ all yield the same quantized measurements, magnitude information (i.e., $\|x\|_2$) is not retrievable. The goal is then to recover only the directional information $x/\|x\|_2$ as accurately as possible. Another challenge associated with one-bit MSQ is tractable decoding. In fact \cite{BB_CISS08}, which initiated this line of work, formulated a recovery algorithm for extracting the direction of sparse signals from one-bit, MSQ quantized, compressive measurements. However this algorithm did not have theoretical recovery guarantees.  Later, a decoder based on convex optimization was proposed for recovering $x/\|x\|_2$ in \cite{PV13} and a rate-distortion relationship \mbox{$\mathcal{D}(\mathcal{R}) \lesssim \mathcal{R}^{-1/5}$} was derived  (here, the distortion is measured using the magnitude-normalized signal and its approximation). Another issue worth mentioning here is that while multi-bit MSQ approaches to quantizing compressed sensing measurements generally allow for sub-Gaussian measurements, the one-bit MSQ setup is different. In particular, general sub-Gaussian measurements in this setting necessitate imposing restrictive assumptions on the signal class (e.g., requiring that the signal not be too sparse  \cite{ai2014one}). Thus, the sensing matrix $\Phi=[\phi_1,...,\phi_m]^T$ is usually restricted to be Gaussian.  

Nevertheless there has been progress in circumventing these issues, for example by deviating slightly from the one-bit MSQ paradigm as outlined in \cite{BB_CISS08}. For example, to circumvent the loss of magnitude information associated with one-bit MSQ, \cite{knudson2014one} added Gaussian (or constant) dither to the measurements so that $$q_i=\textrm{sign}(\langle \phi_i, x \rangle +b_i),$$ where $b_i$ is known. They also proposed techniques for decoding and proved the associated rate-distortion relationship $\mathcal{D}(\mathcal{R})\lesssim \mathcal{R}^{-1/5}$ with magnitude information now accounted for. Finally, we remark that due to the fundamental limitation \eqref{eq:lower_bound_MSQ} associated with MSQ, the exponent $-1/5$ in the rate-distortion relationship can at best be improved to $-1$.

The MSQ limitation \eqref{eq:lower_bound_MSQ}, along with the other issues associated mentioned above, motivated alternative approaches to one-bit quantization in the compressed sensing framework. One such approach, described in \cite{chou2013beta}, introduces a new quantization technique called ``distributed noise-shaping" and obtains a near optimal rate-distortion relationship (in the sense of \eqref{eq:sparse_lower_bound}) with a tractable decoder. The idea here is to replace the difference matrix $D$ in the $\SD$ approach with a block-diagonal matrix $H$. Each block of $H$ is a bidiagonal matrix with $1$ on the diagonal and $-\beta$, where $1 \leq \beta<2$, on the sub-diagonal. Thus, this method can be seen as a generalization of $\beta$-encoding techniques (see, e.g., \cite{daubechies2002beta, daubechies2006robust}). The block-structure of $H$  allows the quantization to be done in a distributed, rather than fully sequential, way.
The obvious advantage of this method is the near-optimal recovery guarantee that it achieves. On the other hand, as a coarse quantization method, this approach requires a large number of analog memory elements (on the order of $k\log N$) to handle additional measurements. 

Yet another one-bit quantization scheme in the compressed sensing setting was proposed in \cite{baraniuk2014exponential}. Here, the idea is to update the quantization thresholds adaptively (as noise-shaping techniques do) albeit by solving a sophisticated convex optimization problem or running a greedy algorithm within the quantization procedure. In particular, \cite{baraniuk2014exponential} proposes the quantization scheme
\begin{equation} q^{(j)} = \textrm{sign}(\Phi^{(j)}(x-x^{(j-1)})-2^{2-j}\tau_j), \label{eq:Bar_it}\end{equation}
where $\Phi^{(j)}$ is a sub-matrix of $\Phi$, $x^{(j-1)}$ is an estimate of $x$ and $\tau_j$ is the current value of the quantization threshold. Here, $x^{(j-1)}$  and $\tau_j$ are obtained by solving an intermediate convex optimization problem. To quantize the measurements of a single vector $x$, one needs to solve many such intermediate problems. Another important issue here is that as $j$ increases, both $\Phi^{(j)}(x-x_{j-1})$ and $2^{2-j}\tau_j$ decrease exponentially fast in $j$. This requires that the physical implementation of the $\textrm{sign}$ function in \eqref{eq:Bar_it} be able to accurately distinguish between very small negative and very small positive quantities. Such ``delicate" comparisons are typical of fine (rather than coarse) quantization schemes 
and are only physically possible up to a certain accuracy. 
In short, for the price of running a polynomial-time algorithm each time the thresholds are updated, \cite{baraniuk2014exponential} achieves exponential error decay in the bit-rate, when the signals are sparse. 

Finally, as noted in Section \ref{sec:SD}, \cite{SWY15} proposed using
  Sigma-Delta ($\Sigma\Delta$) quantization with a subsequent
  reconstruction scheme based on convex optimization. The approach in \cite{SWY15} allows one-bit quantization, provided the $\sd$ scheme is stable. For example, one could use the simple $1$st order greedy $\sd$ scheme in \eqref{eq:SD_greedy} with a one-bit scalar quantizer. One could also use any stable one-bit $r$th order scheme, such as those of \cite{G-exp, DGK10}.
  
 In particular \cite{SWY15}  proves that
  the reconstruction error due to quantization decays polynomially in
  the number of measurements. It  also applies to arbitrary signals,
  including compressible ones,  so it is robust. Moreover, it is stable in the presence of measurement noise.
  This approach, and its associated analysis applies to sub-Gaussian (including Gaussian and
  Bernoulli) random compressed sensing measurements. In this paper we build on \cite{SWY15} and show that by adding an appropriate encoding stage exponential error decay (in the number of bits) can be achieved without sacrificing stability or robustness. Moreover, the results still hold for sub-Gaussian measurements. 


\section{Encoding quantized compressive samples: exponential accuracy}\label{sec:encoding}
We now describe the proposed AIC in detail following the framework and
notation laid out in the Introduction. Thus, our scheme consists of a
compressive sampling stage followed by quantization and encoding
stages. Subsequently, the underlying signal is reconstructed via a
one-stage decoder. 
\medskip

\noindent {\bf Compressive sampling.} We assume that the signal of
interest is $x\in \R^N$. We use an $m\times N$ sub-Gaussian (e.g.,
Gaussian or Bernoulli) compressive sensing matrix $\Phi$. We denote
the rows of $\Phi$ by $\phi_i$ which we view as vectors in $\R^N$.
The resulting (possibly noisy) measurement vector is 
\begin{equation}\label{cs_noisy}
y=\Phi x+e
\end{equation}
with entries
$$y_i=\langle \phi_i,x\rangle+e_i,$$ 
where $1\le i\le
m$, $e$ denotes additive noise, and $|e_i|\le \epsilon$ for a known
$\epsilon\ge 0$. 
\medskip

\noindent{\bf Quantization.} We quantize the compressive measurement
vector $y$ using a stable, $r$th-order $\sd$ scheme (fine or coarse) with alphabet
$\A_\delta^K$ as defined in \eqref{eq:alph} -- see Section~\ref{sec:SD} for details.  
\medskip

\noindent{\bf Encoding.} 
As initially proposed in \cite{IS13} in the context of finite-frames, we encode the $r$-th order $\sd$-quantization $q\in\A^m$ via $$\mathcal{E}: q \mapsto BD^{-r} q.$$ 
Here, $B$ is an $L \times m$ matrix with i.i.d equiprobable Bernoulli random entries and $$m \geq L \geq  c k\log(N/k) $$ 
for an appropriate constant $c$ (see Section \ref{sec:main_results}). Thus, the encoding consists of first multiplying integer-valued (modulo $\delta/2$) vectors  $D^{-r}q$ by a Bernoulli matrix to reduce the dimension, and then assigning a binary label to the result. So, it is easily implementable in the digital domain. In short, it can be seen that the encoding map $\mathcal{E}: \A^m \to \mathcal{C}$ produces codewords in  $\mathcal{C}$ with $\log_2|\mathcal{C}|\ll m\log_2|\mathcal{A}| $ so the encoded measurements to be represented by $\log_2|\mathcal{C}|$ bits instead of the original $m\log_2|\mathcal{A}|$ bits. The goal of the decoder will be to ensure that this compression does not adversely affect the reconstruction quality.

Algorithm~\ref{alg1} illustrates the acquisition side of the proposed AIC. 
 Algorithm~\ref{alg1}  shows that the acquisition side of the AIC is causal
and progressive, i.e., each additional measurement $y_i$ is quantized
without a need to know the ``future'' measurements $y_j$, $j >
i$. Further, during the encoding, each $\sd$ quantized measurement $q_i$ is
incorporated to obtain an updated codeword. This process yields a progressively better
approximation for sufficiently large $i$, as our main theoretical results demonstrate. 

\begin{algorithm}
\renewcommand{\algorithmicrequire}{\textbf{Initialize:}}
\renewcommand{\algorithmicensure}{\textbf{Input:}}
\caption{Acquisition side of the proposed AIC.}
\begin{algorithmic}[1]\label{alg1}
\ENSURE Sub-Gaussian compressive sensing measurement vectors $\phi_i\in \R^N$
\ENSURE Bernoulli encoding matrix $B\in\{\pm 1\}^{L\times m}$ with columns $B_i\in \R^L$
\ENSURE Quantization alphabet $\A$
    \REQUIRE State variables $u_j=0$ and measurements $y_j =0$, for all $j\leq 0$.
    \REQUIRE $L$-dimensional codeword $c = 0$
\FOR {$i=1$ to $m$}
\STATE Obtain the (possibly noisy) compressive measurement: $y_i = \langle \phi_i, x \rangle + e_i$.
\STATE Quantize the measurement and update the state variable, for example using the $1$st order greedy $\sd$ scheme:  \begin{align*}
q_i &= Q_\A\left(y_i + u_{i-1}\right), \notag \\
u_i&=  u_{i+1} + y_i - q_i.
\end{align*} 
Alternatively, use a stable $r$th order scheme as in \eqref{equ:rthOrdu_gen}.
\STATE Update the encoding: 
$c \leftarrow c + B_i q_i$ (note that this is equivalent to setting $c=Bq$).

\ENDFOR
    \end{algorithmic}
\end{algorithm}

\medskip

\noindent{\bf Decoding:} The decoding is done via convex optimization. Specifically, in the absence of non-quantization noise, we compute the estimate
\begin{eqnarray}\label{eq:decoder_m}
& \hat{x} :=  \arg\min\limits_{z}\|z\|_1, \text{ subject to }  \|B D^{-r}(\Phi z-q )\|_2 \leq 3C m,
\end{eqnarray} 
where $C$ is a constant that depends on the $\sd$ quantizer. More generally, in the presence of bounded non-quantization noise $e\in\R^m$ satisfying $\|e\|_\infty\leq \epsilon$, our decoder computes an approximation $\hat{x}$ by solving
\begin{equation}\label{eq:opt}
(\hat{x},\hat{u}, \hat{e})=\arg\min \|\widetilde{x}\|_1, \text{  s.t. } \left\{\begin{array}{rcl} 
BD^{-r}(\Phi \widetilde{x}+\widetilde{e}) -B\widetilde{u} &=&BD^{-r}{q}\\
\|B\widetilde{u}\|_2 &\leq& 3Cm\\
\|\widetilde{e}\|_2 &\leq & \sqrt{m} \varepsilon.
\end{array}\right.
\end{equation}
Intuitively, this decoder makes sense as it produces an estimate that is consistent with the properties of the encoding, the quantization, and the noise. In particular, as we are quantizing the noisy measurements $\Phi x + e$, the stable $\sd$ quantizer will produce $q=\Phi x + e +D^r u $. Applying $BD^{-r}$ on both sides of \eqref{eq:opt} yields the first constraint. Moreover, due to stability of the $\sd$ quantizer we have $\|u\|_\infty \leq C$. This implies that $\|u\|_2 \lesssim \sqrt{m}$ and concentration of measure properties of Bernoulli matrices then yield $\|Bu\|_2 \lesssim m $, hence the second constraint. Finally, the third constraint in \eqref{eq:opt} is a direct consequence of $\|e\|_\infty \leq \epsilon$.

Theorem \ref{thm:main}, our main rate-distortion result, holds for $r\geq 2$, the decoder \eqref{eq:decoder_m}, and \emph{uniformly} for all appropriately bounded signals $x$, but does not cover the case $r=1$. To partially remedy this, we obtain a similar \emph{unifrom} result  (Theorem \ref{thm:impv}) that deals with the $r=1$ case, albeit for strictly sparse signals under noiseless measurements. Here our decoder is modified so that it now obtains an approximation $\hat{x}$ from the encoded measurements of $x$ by solving 
\begin{equation}\label{eq:noiseless}
\hat{x}=\arg\min \|\widetilde{x}\|_1, \text{  s.t. } \|BD^{-r}\Phi \widetilde{x}-BD^{-r} q\|_2\leq (2+\eta) \gamma(r)\sqrt {mL}.
\end{equation}
The constraint in the modified decoder \eqref{eq:noiseless} is based on a series of observations. First, a stable $\sd$ quantizer with stability constant $\gamma(r)$ produces $q$ satisfying $D^{-r}q:= D^{-r}\Phi x - u$ with  $$\|u\|_2 \leq \gamma(r)\sqrt{m}.$$ Second, the random matrix $B$ serves as a Johnson-Lindenstrauss embedding (see Lemma \ref{lm:ls_pts}) so that for a fixed finite set of signals $x$ the associated state variables $u$ satisfy $$\|Bu\|_2 \leq (1+\eta)\gamma(r)\sqrt{mL} 
$$ with high probability (that depends on $\eta$). Passing to arbitrary sparse signals $x$ (i.e., not just from the finite set), we replace the upper bound on $\|Bu\|_2$ by $(2+\eta)\gamma(r)\sqrt{mL}$ (see the proof of Theorem \ref{thm:impv} for the details).

\section{Main Results}\label{sec:main_results}

%
%
\begin{thm}\label{thm:main} Let $\Phi$ be an $m\times N$ 
sub-Gaussian matrix with mean zero and unit variance, and let $B$ be an $L\times m $ Bernoulli matrix with  $\pm 1$ entries. Moreover, let $k\in \{1,...,\min\{m,N\}\}$.  

Denote by $Q_{\sd}^r$ a stable
  $r$th-order scheme with alphabet $\A_\delta^K$ and stability
  constant $\gamma(r)$. There exist positive constants $\Cl{RIPc}$, $\Cl{C_alpha}$, $\Cl{whatever1}$, and $\Cl[littlec]{whatever2}$ such that whenever $\frac{m}{\Cr{C_alpha}}\ge
L \ge \Cr{RIPc} k \log(N/k)$ the following holds with probability
  greater than $1-\Cr{whatever1}e^{-\Cr{whatever2}\sqrt{mL}}$ on the draws of $\Phi$ and $B$: 

Suppose that $x\in \R^N$, $e\in\R^m$ with $\|\Phi x\|_\infty \le \mu < 1$ and that $q:=Q_{\sd}^r(\Phi x+e)$,  where  $\|e\|_\infty\le \epsilon$ for some $0\leq	\epsilon<1-\mu$.  Then the solution $\hat{x}$ to \eqref{eq:opt} satisfies
\begin{equation}\label{eq:err0}
\|x-\hat{x}\|_2 \leq \Cl[errorterms]{d3}\Big(\frac{L}{m}\Big)^{r/2-3/4}\delta + \Cl[errorterms]{d4} \sqrt{\frac{m}{L}}\varepsilon + \Cl[errorterms]{d5} \frac{\sigma_k(x)}{\sqrt{k}}.
\end{equation}
where $\Cr{d3}$, $\Cr{d4}$, and $\Cr{d5}$ are constants.
\end{thm}

\remark{In practice, if the noise vector $e$ is comprised of zero-mean i.i.d. Gaussian random variables with variance $\varepsilon^2$, one can replace the constraint on $\|\widetilde{e}\|_2$ in \eqref{eq:opt} with $\|\widetilde{e}\|_2 \leq  \sqrt L \varepsilon$ and the same proof holds.  By using this new constraint, we obtain an error bound of the form 
\begin{equation}\label{eq:err}
\|x-\hat{x}\|_2 \lesssim \Big(\frac{L}{m}\Big)^{r/2-3/4} + \varepsilon +  \frac{\sigma_k(x)}{\sqrt{k}}.
\end{equation}
In particular, now the reconstruction error due to additive noise is independent of the number of measurements. 
}

\begin{cor}[rate-distortion relationship and exponential error decay]\label{cor:RD} Let $\Phi \in \R^{m\times N}$ and $B\in \R^{L\times m}$ be compressed sensing and encoding matrices, as in Theorem \ref{thm:main}. 
Denote by $Q_{\sd}^r$ a stable $r$th order $\sd$ scheme with alphabet $\A_{\delta}^{K}$, and let $\mathcal{X}:=\{ x\in \R^N: \|\Phi x\|_\infty \le 1\} $. Let all numbered constants be as before. Then, the following hold with probability greater than $1-\Cr{whatever1}e^{-\Cr{whatever2}\sqrt{mL}}$ on the draws of $\Phi$ and $B$:

\begin{itemize}
\item[(i)] The number of bits needed to represent all elements of the codebook  $\mathcal{C} := BD^{-r} \circ Q_{\sd}^r(\mathcal{X})$ is bounded above by 
$$\mathcal{R} = L(r+1)\log_2(m) + L \log_2(2K).$$

\item[(ii)] The resulting rate-distortion relationship associated with decoding by \eqref{eq:opt} is given by
$$\mathcal{D}(\mathcal{R}) \leq \Cl[errorterms]{d_rate} \cdot
2^{-\Cl[littlec]{c_rate} \frac{\mathcal{R}}{L}} +
\Cr{d5}\frac{\sigma_{ \lfloor L/(\Cr{RIPc}\log
      N)\rfloor}}{\sqrt{\lfloor L/(\Cr{RIPc}\log N)}\rfloor},$$
      where 
      $$\Cr{d_rate}:=\Cr{d_rate}(L,r,K,\delta)=\left(\Cr{d3} \cdot L^{r/2-3/4}\cdot \left( 2K \right)^\frac{r/2-3/4}{r+1}\cdot \delta\right),$$
      and 
      $$ \Cr{c_rate}:=\Cr{c_rate}(r)= \frac{r/2-3/4}{r+1}.$$ 

\item[(iii)] In the case of $k$-sparse signals with $k :=\lfloor \frac{L}{\Cr{RIPc}\log N} \rfloor$, we have
$$\mathcal{D}(\mathcal{R}) \leq \Cl[errorterms]{d_rate} \cdot
2^{-\Cl[littlec]{c_rate2} \frac{\mathcal{R}}{k\log N}}$$
where $\Cr{c_rate2}=\Cr{c_rate}/\Cr{RIPc}$.
\end{itemize}
\end{cor}

\begin{cor}[encoding the weak-$\ell_p$ ball]\label{cor:compr} Let $\Phi \in \R^{m\times N}$ and $B\in \R^{L\times m}$ be compressed sensing and encoding matrices, as in Theorem \ref{thm:main}. 
Denote by $Q_{\sd}^r$ a stable $r$th order $\sd$ scheme with alphabet $\A_{\delta}^{K}$, and let 
$B_{w\ell_p}$ be the  weak $\ell_p$ ball of radius $\Cl{C_wlp}$ with $0<p<2$. Then it holds
\begin{equation}\label{eq:wlp_upper}
\mathcal{D}(B_{w\ell_p}) \lesssim  \left(\frac{\log N \log(\mathcal{R}/\log N- \log \log N)}{\mathcal{R}-\log N \log\log N}\right)^{1/p-1/2}. \end{equation}
\end{cor}

\begin{thm}\label{thm:impv} Let $\Phi$ be an $m\times N$ 
sub-Gaussian matrix with mean zero and unit variance and $B$ be an $L\times m $ Bernoulli matrix with  $\pm 1$ entries. Moreover, let $k\in \{1,...,\min\{m,N\}\}$.  Denote by $Q_{\sd}^r$ a stable $r$th-order scheme with alphabet $\A_\delta^K$ and stability constant $\gamma(r)$. There exist positive constants $\Cl{RIPc2}$, $\Cl{C_alpha2}$, $\Cl{whatever12}$, and $\Cl[littlec]{whatever22}$ such that whenever $\frac{m}{\Cr{C_alpha2}}\ge
L \ge \Cr{RIPc2} (k \log N + k\log m)$ the following holds with probability greater than $1-\Cr{whatever12}e^{-\Cr{whatever22}{L}}$ on the draws of $\Phi$ and $B$: 

Suppose that $x\in \Sigma_k^N\cap B_2^N$ with $\|\Phi x\|_\infty  \leq 1$ and that  $q:=Q_{\sd}^r(\Phi x)$. Then the solution $\hat{x}$ to 
\eqref{eq:noiseless}, with $\eta=1$, satisfies
\[
\|\hat{x}-x\|_2 \leq \Cl[errorterms]{c_3}\left( \frac{L}{m} \right)^{r/2-1/4}\delta,
\]
where $\Cr{c_3}$ is a constant.
\end{thm}

\begin{remark}

The above result on first-order $\sd$ quantization holds only for
sparse signals. For arbitrary signals and when the measurements are
noisy, as in \eqref{cs_noisy},
we can obtain a \emph{non-uniform} version of Theorem \ref{thm:main}
to handle $r=1$. That is, the result holds with high probability on
the draw of the encoding matrix, in the regime where one draws a new
random encoding matrix after sensing a fixed finite number of signals.  

\end{remark}

%
%
%
%
%
%
%
%

In this case, we use the decoder 
\begin{equation}\label{eq:opt2}
(\hat{x},\hat{u}, \hat{e})=\arg\min \|\widetilde{x}\|_1, \text{  s.t. } \left\{\begin{array}{rcl} 
BD^{-r}(\Phi \widetilde{x}+\widetilde{e}) -B\widetilde{u} &=&BD^{-r}{q}\\
\|B\widetilde{u}\|_2 &\leq& 2C\sqrt{Lm}\\
\|\widetilde{e}\|_2 &\leq & \sqrt{m} \varepsilon
\end{array}\right.
.\end{equation}
The corresponding approximation sastisfies
\begin{equation}\label{eq:err2}
\|x-\hat{x}\|_2 \leq d_1\Big(\frac{L}{m}\Big)^{r/2-1/4}\delta + d_2 \sqrt{\frac{m}{L}}\varepsilon + \Cr{d3} \frac{\sigma_k(x)}{\sqrt{k}},
\end{equation}
with high probability. The proof of this is essentially the same as that of Theorem
\ref{thm:main} and we omit the details. However, it is worth noting that
the probability with
which \eqref{eq:err2} holds depends on the probability that the true solution $(x,u,e)$ satisfies the
constraint in \eqref{eq:opt2}. This can be calculated by invoking the
Johnson-Lindenstrauss Lemma (Lemma ~\ref{lm:ls_pts})  as $B$ and $u$ are
independent.

\section{Proof of Theorem \ref{thm:main} and its corollaries}\label{sec:proof1}
To prove Theorem \ref{thm:main}, we will require some results from the literature, which we now present.

\subsection{Preliminaries}\label{sec:preliminaries}
We begin with a lemma from \cite{Foucart13}, which will be important for our analysis. 
\begin{prop}[\cite{Foucart13}]\label{pro:foucart} Let $f,g\in \mathbb{C}^N$, and $\Phi \in \mathbb{C}^{m,N}$. Suppose that $\Phi$ has $\delta_{2k}$-RIP with $\delta_{2k}<1/9$, then for any $1\leq p\leq 2$, we have
\[
\|f-g\|_p\leq C_1k^{1/p-1/2} \|\Phi(f-g)\|_2+\frac{C_2}{k^{1-1/p}}(\|f\|_1-\|g\|_1+2\sigma_k(g)_1),
\]
with constants $C_1$, $C_2$ only depend on $\delta_{2k}$.
\end{prop}

Next we present a lemma from \cite{KSY13} essentially bounding, from below, the smallest singular vector of an anisotropic random matrix. We use the lemma to deduce a corollary about the $L$th singular value of the $L\times m$ matrix $BD^{-r}$, which will be useful in the analysis of the encoding scheme. 
\begin{lm}[\cite{KSY13}]\label{lm:sg} let $E$ be an $m\times k$ sub-Gaussain matrix with mean zero, unit variance, and parameter $c$, let $S=\text{diag}(s)$ be a diagonal matrix, and let $V$ be an orthonormal matrix, both of size $m\times m$. Further, let $r\in \mathbb{Z}^+$ and suppose that $s_j\geq \Cl{C_1}^r\left(\frac{m}{j}\right)^r$, where $\Cr{C_1}$ is a positive constant that may depend on $r$. Then there exist constants $\Cl{C_2},\Cl[littlec]{C_3}>0$ (depending on $c$ and $\Cr{C_1}$) such that for $0<\alpha <1$ and $\lambda:=\frac{m}{k} \geq \Cr{C_2}^{\frac{1}{1-\alpha}} $
\[
\mathbb{P} \left(\sigma_{\min} \left(\frac{1}{\sqrt m} SV^*E\right) \leq \lambda^{\alpha(r-1/2)}\right)\leq  2 \exp(-\Cr{C_3} m^{1-\alpha} k^\alpha).
\]
\end{lm}
Lemma \ref{lm:sg} implies the following fact.

\begin{cor}\label{cor:sinvalue}
 Let $D$ be the $m\times m$ difference matrix, and $B$ be an $L\times m$ $(L\leq m)$ Bernoulli random matrix. Let $k\leq L$ be such that ${\Cr{C_2}^\frac{1}{1-\alpha}}k \leq {m}$ where $\Cr{C_2}$ and $\alpha$ are as in Lemma \ref{lm:sg}. Then with probability at least 
 $1-2e^{-\Cr{C_3} m^{1-\alpha} k^\alpha}$
$$
\sigma_k(BD^{-r})\geq \sqrt m\left(\frac{m}{k}\right)^{\alpha(r-1/2)}.
$$
\end{cor}

\begin{proof}
Note that $\sigma_k(BD^{-r}) = \sigma_k((D^*)^{-r}B^*) = \sigma_k(SV^*B^*)$ where $USV^* = (D^*)^{-r}$ is the singular value decomposition of $(D^*)^{-r}$. By Proposition 3.2 of \cite{GLPSY}, we have that the diagonal entries $s_j$ of the diagonal matrix $S$ satisfy $s_j\geq C_1^r\left(\frac{m}{j}\right)^r$ for a (known) constant $C_1$. The result then follows by Lemma \ref{lm:sg}.
\end{proof}

\begin{lm}[corollary of Proposition 4.1 in \cite{KSY13}] \label{lm:RIP} 
Let $V $ be an $m \times L$ $(m\geq L)$ orthonormal matrix and $\Phi $ be an $m\times N$ $(N\geq m)$ sub-Gaussian matrix with mean zero, unit variance, and parameter $c$. With probability over $1-e^{-\Cl[littlec]{C1_KSY}L}$, the matrix $\frac{1}{\sqrt L} V^T \Phi $ satisfies the RIP of order $2k$ and constant $\delta_{2k}$, provided that 
\begin{equation}\label{eq:L} 
L\geq \frac{1}{\delta_{2k}^2}\Cl{C2_KSY}k\log N. 
\end{equation} Here $\Cr{C1_KSY}$ and $\Cr{C2_KSY}$ are absolute constants.
\end{lm}
\begin{proof}
By Theorem 3.7 and Equation (19) of \cite{KSY13},  we can assert that there exists $\Cl[littlec]{C3_KSY}$ such that for any fixed index set $T \subseteq [N]$ with cardinality $k$, we have
\[
P\left(\sup\limits_{y\in S^{k-1} } \left| \|\frac{1}{\sqrt m} V^T\Phi_T y\|_2^2-\mathbb{E} \|\frac{1}{\sqrt m} V^T\Phi_T y\|_2^2\right|\geq \delta \frac{L}{2m}\right ) \leq e^{-\Cr{C3_KSY}\delta^2 L}.
\]
Inserting $\mathbb{E} \| V^T\Phi_T y\|_2^2 = L\|y\|_2^2$ into the above equation and rescaling each term inside the probability, we obtain
\[
P\left(\sup\limits_{y\in S^{k-1} } \left| \|\frac{1}{\sqrt L} V^T\Phi_T y\|_2^2-1\right|\geq \delta /2 \right) \leq e^{-\Cr{C3_KSY}\delta^2 L}.
\]
Under the condition that $L\geq \frac{\Cr{C2_KSY}}{ \delta^2} k\log N$ with some large enough constant $\Cr{C2_KSY}>0$ independent of $\delta, L, m$, and  $N$, a union bound over $N\choose k$ $k$-dimensional subspaces gives
\[
 1-\delta/2 \leq \| \frac{1}{\sqrt L} V^T\Phi_T z\|_2^2  \leq 1+ \delta/2,  \quad \textrm{for all $z\in \R^k$ and $T \subset [N]$ with $|T|=k$ },
\]
 with probability over $1-e^{-\Cr{C1_KSY}L} $. Rescaling $\delta$ completes the argument.
\end{proof}

\subsection{Proofs of main results}

%
%

\begin{proof}[\textbf{ Proof of Theorem \ref{thm:main}}]
By hypothesis, $\|\Phi x + e\|_\infty \leq \mu + \varepsilon \leq 1$. 
This guarantees that the $r$th order $\sd$ quantization is stable with  stability constant $C:=\gamma(r)=\tilde\gamma(r)\delta$. 
Let $$\E_1= \{B\in Bern(L,m): \sigma_L
(BD^{-r})\geq \left(\frac{m}{L}\right)^{r/2 -1/4}\sqrt{m}\},$$ and
$$\E_2:=\{B\in Bern(L,m): \|B\|_{\ell_2\to\ell_2} \le \sqrt{L}+2\sqrt{m}\}.$$  By Corollary \ref{cor:sinvalue}, with $\Cr{C_alpha}=\Cr{C_2}^\frac{1}{1-\alpha}$ (and $L$ in place of $k$), we have $${P}(\E_1)\geq 1-2e^{-\Cr{C_3}
  (mL)^{1/2}}.$$ Furthermore, $${P}(\E_2)\ge 1-\Cl{Feld1}
e^{-\Cl[littlec]{Feld2} L}$$ for some constants $\Cr{Feld1}$ and $\Cr{Feld2}$ by \cite [Corollary V.2.1 with $\epsilon=1$]{feldheim2010}.  
 Setting $\E=\E_1\cap \E_2 $, we get $$\mathbb{P}(\E)\geq 1-2e^{-\Cr{C_3}(mL)^{1/2}}-\Cr{Feld1}e^{-\Cr{Feld2}L} .$$ Now, we note that for any $B\in \E$, by the constraints in \eqref{eq:opt} and the fact that $$\|Bu\|_2 \leq \|B\|_{op} \|u\|_\infty \sqrt{m} \leq 3Cm,$$ we have
 \[
\|BD^{-r}\big(\Phi (\hat{x}-x)+(\hat{e}-e)\big)\|_2 = \|B(\hat{u}-u) \|_2 \leq  \|B\hat{u}\|_2 + \|B{u}\|_2 \leq 6Cm.
\]
Let $BD^{-r}=TSR^T$ be the singular value decomposition of $BD^{-r}$, with the diagonal entries of $S$ arranged in decreasing order, and define $\widetilde{\Phi} = R^T\Phi$. Moreover, let $h=x-\hat{x}$ and $v=e-\hat{e}$. Then
\begin{align*}
6C m& \geq \|BD^{-r}(\Phi h+v)\|_2 = \|TSR^T (\Phi h+v) \|_2 = \|S\widetilde{\Phi} h+SR^Tv\|_2, 
\end{align*}
where the last equality is by unitary invariance of the norm. Denoting, for the moment,
by $A_L$ the restriction of a matrix $A$ to its first $L$ rows, we have 
\begin{align*}
6Cm & \geq \|S_L \widetilde{\Phi} h + S_L R^T v\|_2 \\
&\geq \sigma_L(S)\sqrt L\|\frac{1}{\sqrt L}(\widetilde{\Phi}_Lh + (R^T)_{{}_L }v)\|_2  \\&= \sigma_L(BD^{-r})\sqrt L\|\frac{1}{\sqrt L}(\widetilde{\Phi}_Lh + (R^T)_{{}_L} v)\|_2\\ 
&\geq \sqrt m \left(\frac{m}{L}\right)^{r/2-1/4}\sqrt L \|\frac{1}{\sqrt L}(\widetilde{\Phi}_Lh + (R^T)_{{}_L}  v\|_2.
\end{align*}
Above, for the second inequality we used the fact that $S$ is diagonal with its diagonal elements in decreasing order. For the last inequality we used the fact that $B\in \E$ to bound $\sigma_L(BD^{-r})$. Rearranging and using the reverse triangle inequality,
\[
\|\frac{1}{\sqrt L} \widetilde{\Phi}_Lh\|_2\leq 6C \Big(\frac{L}{m}\Big)^{r/2-3/4} + \frac{1}{\sqrt{L}}\|(R^T)_{{}_L} v\|_2. 
\]
Now, using the fact that \begin{equation}\|(R^T)_{{}_L} v\|_2\leq \| R^T v\|_2 = \|v\|_2 \leq 2\sqrt{m} \varepsilon\label{eq:v_bound},\end{equation}  we deduce that
\[
\|\frac{1}{\sqrt L} \widetilde{\Phi}_Lh\|_2\leq 6C \Big(\frac{L}{m}\Big)^{r/2-3/4} + 2\sqrt{\frac{m}{L}}\varepsilon.
\]
Let $\mathcal{E}_3$ be the event that $\frac{1}{\sqrt L} \widetilde{\Phi}_L$ satisfies the RIP of order $2k$ 
with constant $\delta_{2k}<1/9$. Then $$\mathbb{P}(\mathcal{E}_3) \geq 1-e^{-\Cr{C1_KSY}L}$$ by Lemma~\ref{lm:RIP}. 
Applying Proposition \ref{pro:foucart} with $p=2$, $\frac{1}{\sqrt
  L}\widetilde{\Phi}_L$ in place of $\Phi$, and $x$ and
$\hat{x}$ in place of $g$ and $f$ respectively, we obtain
\begin{equation}\label{eq:err}
\|x-\hat{x}\|_2 \leq \Cr{d3}\Big(\frac{L}{m}\Big)^{r/2-3/4}\delta + \Cr{d4} \sqrt{\frac{m}{L}}\varepsilon + \Cr{d5} \frac{\sigma_k(x)}{\sqrt{k}}.
\end{equation}
Thus \eqref{eq:err0} holds with probability $\mathbb{P}(\mathcal{E}_1 \cap \mathcal{E}_2 \cap \mathcal{E}_3) \geq 1-\Cr{whatever1}e^{-\Cr{whatever2}\sqrt{mL}}$.
\end{proof}

\proof[\textbf{Proof of Corollary~\ref{cor:RD}}]
We start with calculating the number of bits needed to represent a codeword $\tilde{q}:= BD^{-r}q  = BD^{-r}Q_{\sd}^r(x)$. The entries of $q$ take on values in $\A_\delta^K$, with $|\A_\delta^{K}| = 2K$, as seen from the definition \eqref{eq:alph}. 
As $\|D^{-r}\|_{\ell_\infty\rightarrow
  \ell_\infty} \leq m^r$ and $\|B\|_{\ell_\infty\rightarrow \ell_\infty}\leq
m$, we have $\|\widetilde{q}\|_\infty =\|BD^{-r}q\|_\infty\leq
m^{r+1}\|q\|_\infty$. Thus, recalling the definition of $\A_\delta^K$, it can be seen that each entry of $\widetilde{q}$ takes on values in a set of cardinality at most $2m^{r+1}K$. We have $L$ such entries and therefore need at most \begin{equation}\mathcal{R} = L(r+1)\log_2(m) + L \log_2(2K)\label{eq:OurRate}\end{equation} bits to uniquely represent $\widetilde{q}$.

Next, to control the distortion we apply Theorem \ref{thm:main} with $\epsilon=0$ and with $k,L,$ and $m$ satisfying
$$\frac{m}{\Cr{C_alpha}}\ge L \ge \Cr{RIPc} k \log (N/k).$$ 
In particular, we choose  $$k = \lfloor
L/(\Cr{RIPc} \log N) \rfloor.$$

This gives an upper bound on the reconstruction error associated with \eqref{eq:opt}, namely
\[
\|\hat{x}-x\|_2 \leq \Cr{d3} \left(\fr{L}{m}\right)^{r/2-3/4}\delta+\Cr{d5}\frac{\sigma_k(x)}{\sqrt{k}}=:\mathcal{D}.
\]
Solving for $m$ in \eqref{eq:OurRate} and substituting in the distortion expression above, we have the rate-distortion relationship 
$$\mathcal{D}(\mathcal{R}) \leq \left(\Cr{d3} \cdot L^{r/2-3/4}\cdot \left( 2K \right)^\frac{r/2-3/4}{r+1}\cdot \delta\right) \cdot
2^{-\frac{r/2-3/4}{r+1} \frac{\mathcal{R}}{L}} +
\Cr{d5}\frac{\sigma_{ k}(x)}{\sqrt{k}}.$$ Substituting for $k$ completes the proof of (ii). Finally (iii) follows trivially, since for $k$-sparse signals $\sigma_k(x)=0$ and since $L\leq \Cr{RIPc} k \log N$.
\endproof

\begin{proof}[\textbf{ Proof of Corollary~\ref{cor:compr}}]
By Corollary \ref{cor:RD}, the number of bits needed to represent $BD^{-r}q$ is \begin{equation}\mathcal{R} = L(r+1)\log_2(m) + L \log_2(2K) = L \log_2(2m^{r+1}K).\label{eq:rate1}\end{equation} 
Moreover, the resulting approximation error from Theorem \ref{thm:main} is \[
\|\hat{x}-x\|_2 \leq \Cr{d3} \left(\fr{L}{m}\right)^{r/2-3/4}\delta+\Cr{d5}\frac{\sigma_k(x)}{\sqrt{k}}.
\]
Take any $x \in \mathcal{X}$ and assume without loss of generality that  the entries of $x$ are sorted in decreasing order of magnitude. Then, we have 
\begin{align*}
\frac{\sigma_k(x)}{\sqrt{k}} &= \frac{1}{\sqrt{k}} \sum_{j=k+1}^N |x_j| \\ 
		           		 &\leq \frac{\Cr{C_wlp}}{\sqrt{k}} \sum_{j=k+1}^N j^{-1/p}
					    \leq \frac{\Cr{C_wlp}}{\sqrt{k}} \int_k^\infty z^{-1/p}dz \\
					 &\leq \frac{\Cr{C_wlp}}{1/p-1}k^{1/2-1/p}. 
\end{align*} 
Defining $\Cl[errorterms]{Ccombined}:=\Cr{Ccombined}(p)=\max(\Cr{d3}\delta,\Cr{d5}\frac{\Cr{C_wlp}}{1/p-1})$ we now have 
\[\|x-\hat{x}\|_2 \leq \Cr{Ccombined} \left( \left(\frac{L}{m}\right)^{r/2-3/4} + k^{1/2-1/p} \right).\]
Noting that the right hand side of above holds for any $k\leq \frac{L}{\Cr{RIPc}\log N}$ (once $\Cr{Ccombined}$ is replaced by its maximum over all $k\leq \frac{L}{\Cr{RIPc}\log N}$ if necessary), we select $k=\lceil\frac{L}{2\Cr{RIPc}\log N}\rceil \in [ \frac{L}{2\Cr{RIPc}\log N},  \frac{L}{\Cr{RIPc}\log N}]$. Moreover, since the same right hand side is a decreasing function of $k$, we may remove the ``ceiling" function to obtain 
\[\|x-\hat{x}\|_2  \leq \Cr{Ccombined} \left( \left(\frac{L}{m}\right)^{r/2-3/4} + \left(\frac{L}{2\Cr{RIPc}\log N}\right)^{1/2-1/p} \right). \]
Setting the two summands in the right hand side above to be equal, we have
\begin{equation}\label{eq:mL}
m = \left(\frac{L^{r/2-5/4+1/p}}{(2\Cr{RIPc}\log N)^{1/p-1/2}}\right)^\frac{1}{r/2-3/4}.
\end{equation}
Note that since $L\geq 2\Cr{RIPc} \log N$, then $m\geq 2\Cr{RIPc}\log N$ follows, and this choice of $m$ yields
\begin{align}
\|x-\hat{x}\|_2 &\leq 2 \Cr{Ccombined} \left(  \frac{L}{2\Cr{RIPc}\log N}\right)^{1/2-1/p}=:\mathcal{D}\nonumber\\
\mathcal{R}&= L\log_2\left(\frac{ 2 L^{(r/2-5/4+1/p)(r+1)/(r/2-3/4)}}{(2\Cr{RIPc}\log N)^{(1/p-1/2)(r+1)/(r/2-3/4)}}K\right).
\nonumber\end{align}
Upon rearranging, we have $L \leq(2\Cr{RIPc}\log N)(\frac{\mathcal{D}}{2\Cr{Ccombined}})^\frac{1}{1/2-1/p}$, which when substituted into the expression for $\mathcal{R}$ gives

\begin{align}
 \mathcal{R} &\leq  (2\Cr{RIPc}\log N) \left( (r+1)\log_2(2\Cr{RIPc}\log N) + \log_2 K + \log_2\left(\left(\frac{2\Cr{Ccombined}}{\mathcal{D}}\right)^{\frac{r+1}{1/p-1/2}+\frac{r+1}{r/2-3/4}} \right) \right) \left(\frac{2\Cr{Ccombined}}{\mathcal{D}}\right)^\frac{1}{1/p-1/2}
 \\ & \leq \Cl{Compr1} \log N \log\log N + \Cl{Compr2} \log N \left(\frac{2\Cr{Ccombined}}{\mathcal{D}}\right)^\frac{1}{1/p-1/2}\log\left(\left(\frac{2\Cr{Ccombined}}{\mathcal{D}}\right)^{\frac{r+1}{1/p-1/2}+\frac{r+1}{r/2-3/4}}\right).
 \end{align}
For large enough $\mathcal{R}$, say it is such that $a:=\mathcal{R}-\Cr{Compr1}\log N \log \log N>1$ and $b:=\frac{2\Cr{Ccombined}}{\mathcal{D}}>1$,  we can use the fact that $a\leq c b\log b$ implies $a \leq  c b \max\{\log a/c,1\}$ when $a,b>1$ and $c$ is a positive constant.  Then we obtain 
$$D\lesssim \left(\frac{\log N \log(\mathcal{R}/\log N- \log \log N)}{\mathcal{R}-\log N \log\log N}\right)^{1/p-1/2}.$$
			
\end{proof}

\section{Proof of Theorem \ref{thm:impv}}\label{sec:proof2}

%

First, we state the Johnson-Lindenstrauss Lemma \cite{JLoriginal}, in the form it appears in \cite{baraniuk2008simple}.
\begin{lm}[Johnson-Lindenstrauss lemma]\label{lm:ls_pts} Suppose $\eta$, $p\in (0,1)$, and $S$ is a finite set in $\mathbb{R}^m$. Let $B\in \R^{L\times m}$ be a Bernoulli random matrix whose entries take value 1 or -1 with equal probability. Set $\widetilde{B}=\frac{1}{\sqrt L} B$. Then 
$$
(1-\eta)\|x\|_2 \leq \|\widetilde{B} x\|_2\leq (1+\eta)\|x\|_2
$$
for all $x\in S$ with probability at least $1-p$, provided that 
provided that $L\geq \frac{4+2\log_{|S|}(1/p)}{\eta^2/2-\eta^3/3}\ln|S|.$
\end{lm}
\begin{defn}[Quantization cell] For a fixed  quantizer $Q: \mathbb{R}^m \rightarrow \mathcal {A}^m $ and a measurement scheme 
\begin{align*}
E: & \mathcal {X} \rightarrow \mathbb{R}^m \\ &x\rightarrow \ \ y,
\end{align*}
the quantization cells associated with $Q$ and $E$, that intersect $\mathcal{X}$, are defined by
$$C_{Q,E}(q) =\{ x\in \mathcal {X}: \ \ Q(E(x))= q\},$$
where we call $q \in \mathcal{A}^m$ the center of the cell $C_{Q,A}(q)$.
\end{defn}

The following lemma provides an upper bound on the number of $\sd$ cells in a bounded region. It is due to Sinan G{\"u}nt{\"u}rk. We provide his original (unpublished) proof in the appendix.   
\begin{lm}\label{lm:sinan} Let $B_R$ be the $k$ dimensional $\ell_2$ ball with radius $R$, $E$ an $m\times k$ $(k\leq m)$ sub-Gaussian matrix. Denote by $Q_{\SD}^r$ the $r$th order stable $\sd$ quantizer and assume the quantization step size is $\delta$. For any $\alpha>0$, with probability $1-e^{-\alpha^2 mk}$, the total number $N_C$ of $\sd$ cells associated with $Q_{\sd}^{r}$ and $E$, that intersect $B_R$, satisfies
\begin{equation}\label{eq:n_cells}
N_C \leq \Cl{C_r}  \max\left\{2^{2k}, \delta^{-k} R^k (1+\alpha)^k\frac{m^{(r+1)k}}{k^{k/2-1/2}}\right\},
\end{equation}
where $\Cr{C_r}$ is a positive constant that depends only on $r$ and the parameters of the sub-Gaussian distribution.
\end{lm}
\begin{lm}\label{lm:constraint}
Let $B \in \mathbb{R}^{L\times m}$ be a Bernoulli random matrix, $E \in \mathbb{R}^{m \times k}$ be a sub-Gaussian matrix, and $D \in \mathbb{R}^{m \times m}$ be the difference matrix.  Let $B_2^k$ be the unit ball of $\ell_2$ in $\mathbb{R}^k$. Assume that 
\begin{equation}\label{eq:Lbound}
L \geq \Cl{c_1}\frac{\log(\frac{1}{\theta} )+k\log m}{\eta^2/2-\eta^3/3}.
\end{equation}  Then for any $x\in B_2^k$, with probability over $1-e^{- mk}-\theta $, we have
\begin{equation}\label{eq:constraint}
\|BD^{-r}(E x-q)\|_2 \leq \sqrt{m L}(2+\eta),
\end{equation}
where $q=Q^r_{\Sigma\Delta}(E x)$ and $Q^r_{\Sigma\Delta}$ denotes a stable $r$th order $\sd$ quantizer with step size $\delta=1$, and $\Cr{c_1}$ is a positive constant that may depend on $r$. 
\end{lm}
\begin{proof}

A standard $\epsilon$-net, say $S$, for $B_2^k$ is a finite set of points in $B_2^k$ with the property that for each  $x\in B_2^N$ there exists an $s\in S$ with $\|x-s\|_2 \leq \epsilon$. We will work with such a net $S$ for $B_2^k$, but we require that in addition to 
$s$ being $\epsilon$-close to $x$ that it lies in the same quantization cell as $x$. That is 
 \begin{equation}\label{eq:cell-net}\max\limits_{x\in B_2^k} (\ \  \min\limits_{s\in S, Q_{\sd}^r (Ex)=Q_{\sd}^r (Es)}\|x-s\|_2) \ \ \leq \ \ \epsilon.
 \end{equation} 
 
In particular, if such an $S$ exists with a sufficiently small $\epsilon$, say $\epsilon =\sqrt{mL}/\|BD^{-r}E\|_2$, and if all the points in $S$ satisfy a version of \eqref{eq:constraint}  with a slightly tighter bound,
 \begin{equation}\label{eq:constraint_loose}
 \|BD^{-r}(Es-q)\|_2\leq \sqrt{mL}(1+\eta),
 \end{equation}
  then we can show that \eqref{eq:constraint} is satisfied by all the points in $B_2^k$.  Indeed, suppose \eqref{eq:cell-net} and \eqref{eq:constraint_loose} hold. Let $x\in B_2^k$, and let $s$ be its associated point in $S$. Denote  
  $q_x=Q_{\sd}^r(Ex)$, and 
  $q_s=Q_{\sd}^r(Es)$, then using the triangle inequality we have
\begin{align*}\label{eq:decomp}
\|BD^{-r}(Ex-q_x)\|_2&\leq \|BD^{-r}(Es-q_s)\|_2+\|BD^{-r}E(s-x)\|_2  +\|BD^{-r}(q_x-q_s)\|_2 \\&\leq (1+\eta)\sqrt{mL}+\|BD^{-r}E\|_2\epsilon \\
&\leq \sqrt{mL}(2+\eta).
\end{align*} 
Above, for the second inequality we  used the assumption that $s$ satisfies \eqref{eq:constraint_loose} and that $s$ lies in the same cell as $x$.
The remaining task is to find an $\epsilon$-net $S$ that satisfies both \eqref{eq:cell-net} with $\epsilon =\sqrt{mL}/\|BD^{-r}E\|_2$ and \eqref{eq:constraint_loose}.

To cover $B_2^k$ in the sense of \eqref{eq:cell-net}, we first use $\epsilon$-balls to cover each quantization cell that intersects with $B_2^k$, and then put them together to get a cover of $B_2^k$. Since $C_{Q_{\sd}^r,E} \cap B_2^k \subseteq B_2^k$, the number of $\epsilon$ balls to cover a single cell is simply bounded by the cardinality of a net for $B_2^k$ and that is, in turn,  bounded by $\left(\frac{3}{\epsilon}\right)^k$ (see, e.g., \cite{foucart2013mathematical}). Moreover, by Lemma \ref{lm:sinan}, the total number of cells, $N_C$, is bounded above by \eqref{eq:n_cells}. Set $\alpha =1$ in \eqref{eq:n_cells}. Then with probability exceeding $1-e^{-mk}$ on the draw of $E$, the cardinality of $S$ can be bounded by
\[
|S| \leq \left(\frac{3}{\epsilon} \right)^k N_C \leq \Cl{C_r2}  m^{4kr}.
\]  
Here $\Cr{C_r2}$ is a positive constants that depend on $r$. To obtain the above bound, we condition on the event $\|E\|_2 \leq \|E\|_F\leq 2\sqrt{mk}$ (which is the same event that yields the probability bound in Lemma \ref{lm:sinan}, hence we don't need to account for it again), and we use the estimates 
 $$\epsilon = \frac{\sqrt{mL}}{\|BD^{-r}E\|_2} \geq \frac{\sqrt{mL}}{\|B\|_2\|D^{-r}\|_2\|E\|_2}  \geq \frac{\sqrt{mL}}{\sqrt{mL} \cdot m^r \cdot 2\sqrt{km}}  \geq \frac{1}{2} m^{-r-1}.$$
Now for this $S$, we use the stability of the $\sd$ quantizer, i.e.,  $\|D^{-r}(Es-q)\|_2\leq \gamma(r)\sqrt m$ and the Johnson-Lindenstrauss Lemma (Lemma  \ref{lm:ls_pts}), to get \eqref{eq:constraint} satisfied  by all points in $S$ with probability $\theta$ as long as $L$ satisfies \eqref{eq:Lbound}.  Hence both \eqref{eq:cell-net} and \eqref{eq:constraint_loose} are now satisfied as desired, and the proof is complete. In particular, the probability is obtained by combining $\theta$ with the the probability of failure associated with Lemma \ref{lm:sinan}.
\end{proof}
\begin{remark} For a general step size $\delta$ of the alphabet, \eqref{eq:constraint} generalizes to 
\[
\|BD^{-r}(E x-q)\|_2 \leq \sqrt{m L}(2+\eta)\delta,
\]
provided that 
$L \geq \Cr{c_1} \frac{\log(\frac{1}{\theta} )+k\log \frac{m}{\delta}}{\eta^2/2-\eta^3/3},$ by some minor modifications to the proof of Lemma \ref{lm:constraint}. 
\end{remark}

\begin{proof}[\textbf{Proof of Theorem \ref{thm:impv}}]
Fix a support set $T\subset \{1,...,N\}$ with $|T|=k$ and invoke Lemma \ref{lm:constraint} with $\theta =e^{-2k\log (eN/k)}$. Then, provided  $$L\geq \Cr{c_1} \frac{k\big(2\log( eN/k)+ \log m \big)}{\eta^2/2-\eta^3/3},$$ all $x\in \Sigma_k^N\cap B_2^N$ that are supported on $T$ satisfy \eqref{eq:constraint} with probability exceeding  $1-e^{-2k\log(eN/k)}-e^{-mk}$. Applying a union bound over all sets $T$ with $|T|=k$, we conclude that the probability that there exists $x\in \Sigma_k^N \cap B_2^N$ violating \eqref{eq:constraint} is  no greater than $\binom{N}{k} (e^{-2k\log (eN/k)}+e^{-km}) \leq 
e^{- k\log(eN/k)}+e^{-km+k\log (eN/k)}$ by using the bound $\binom{N}{k}\leq \big(\frac{eN}{k}\big)^k$. 
Conditioning on the event that \eqref{eq:constraint} is satisfied for all $x \in \Sigma_k^N\cap B_2^N$, we have
\begin{equation}\|BD^{-r} \Phi (\hat{x}-x)\|_2 \leq 2\delta\sqrt{mL}(2+\eta).\label{eq:const1}\end{equation}
 On the other hand, let $BD^{-r}=TSR^T$ be the singular value decomposition of $BD^{-r}$. By Lemma \ref{lm:RIP}, there exists a $\Cr{C1_KSY}$ such that as long as $L\geq \Cr{C2_KSY}\frac{1}{\delta_{2k}^2} k\log N $,  the matrix $\frac{1}{\sqrt L}R^T\Phi$ satisfies the RIP of order $2k$ with (say) constant $\delta_{2k}=1/10$ with probability over $1-e^{-\Cr{C1_KSY} L}$. Conditioning further on this event, we have 
\begin{align}\label{eq:const2}
\|BD^{-r} \Phi (\hat{x}-x)\|_2& = \|SR^T\Phi (\hat{x}-x)\|_2 \geq \sigma_L(BD^{-r}) \| R^T\Phi (\hat{x}-x)\|_2  \nonumber\\
& \geq \sqrt m\left(\frac{m}{L}\right)^{r/2-1/4}\sqrt L \|\frac{1}{\sqrt L}R^T\Phi(\hat{x}-x)\|_2.
\end{align}
Combining the inequalities \eqref{eq:const1} and \eqref{eq:const2}, we get
\[
\|\frac{1}{\sqrt L}R^T \Phi (\hat{x}-x)\|_2 \leq 2\delta(2+\eta)\left(\frac{L}{m}\right)^{r/2-1/4}.
\]
Finally, set $\eta=1$, and apply Proposition \ref{pro:foucart}  to complete the proof.
\end{proof}


\section{Numerical Experiments}
\noindent\textbf{Experiment 1: {Sparse signals}}.
We illustrate the rate-distortion relation \eqref{expdecaycomp}, which was derived as a corollary of the stability results in Theorem \ref{thm:main} and Theorem \ref{thm:impv}. We set $N=1200$, the sparsity level $k=5$, and we vary $m$ in the interval $[m_{\min},m_{\max}]$ with $m_{\min}=\lfloor 10^{2.1}\rfloor $ and $m_{\max}= 10^{3} $. We randomly generate and fix an $m_{\max}\times N$  Gaussian measurement matrix $\Phi $ and an $L \times m_{\max}$ Bernoulli encoding matrix $B$, with $L=200$. By Theorem \ref{thm:main}, fixing $L$ in this way is permitted, as  $k$ and $N$ are both fixed.    
For each value of $m$, we form the sensing matrix $\Phi_m$ using the first $m$ rows of $\Phi$ and the JL matrix $B_m$ using the first $m$ columns of $B$. A test set $S=\{x_i \in \Sigma_5^N: i=1,...,T\}$, with $T=50$ sparse signals is generated. Here, each non-zero entry of $x_i$ is generated independently from the standard Gaussian distribution. Denote by $\hat{x}_i^1$ be the reconstruction from $q_i^1=Q_{\sd}^1( \Phi_m x_i)$ via \eqref{eq:opt2} and $\hat{x}_i^2$ be the reconstruction from $q_i^2=Q_{\sd}^2(\Phi_m x_i)$ via \eqref{eq:opt}, both with $\epsilon=0$. Note that here we use the quantization alphabet $\mathcal{A}_{\delta}^K$ with $K=20$, $\delta=0.1$.  Figure \ref{Fig:sparse} is a  semi-log plot of the average error $\mathcal{D}=\frac{1}{T} \sum\limits_{i=1}^T \|x_i- \hat{x}_i^r\|^2_2$ of the $r$th order $\sd$ ($r=1,2$) as a function of the rate $\mathcal{R}=L\log_2 2m^{r+1} K $ (as computed in Corollary \ref{cor:RD}). 
The slopes of the linear fitting curves match the coefficient in front of $\mathcal{R}$ in  \eqref{expdecaycomp}. This is because when $L$ is known and the signal is strictly sparse, \eqref{expdecaycomp}  reduces to 
\[
\mathcal{D}\lesssim 2^{-\frac{(r/2-3/4)\mathcal{R}}{L}}
\]
if \eqref{eq:opt} is used for reconstruction and
\[
\mathcal{D}\lesssim 2^{-\frac{(r/2-1/4)\mathcal{R}}{L}}
\]
if \eqref{eq:opt2} is used for reconstruction.  This implies that,  for the current choice of parameters, the theoretical slopes for both curves are about $0.00125$ which is close to what we observe in the figure.

\begin{figure}[htbp]
\begin{center}
\includegraphics[width=9cm, height=7cm]{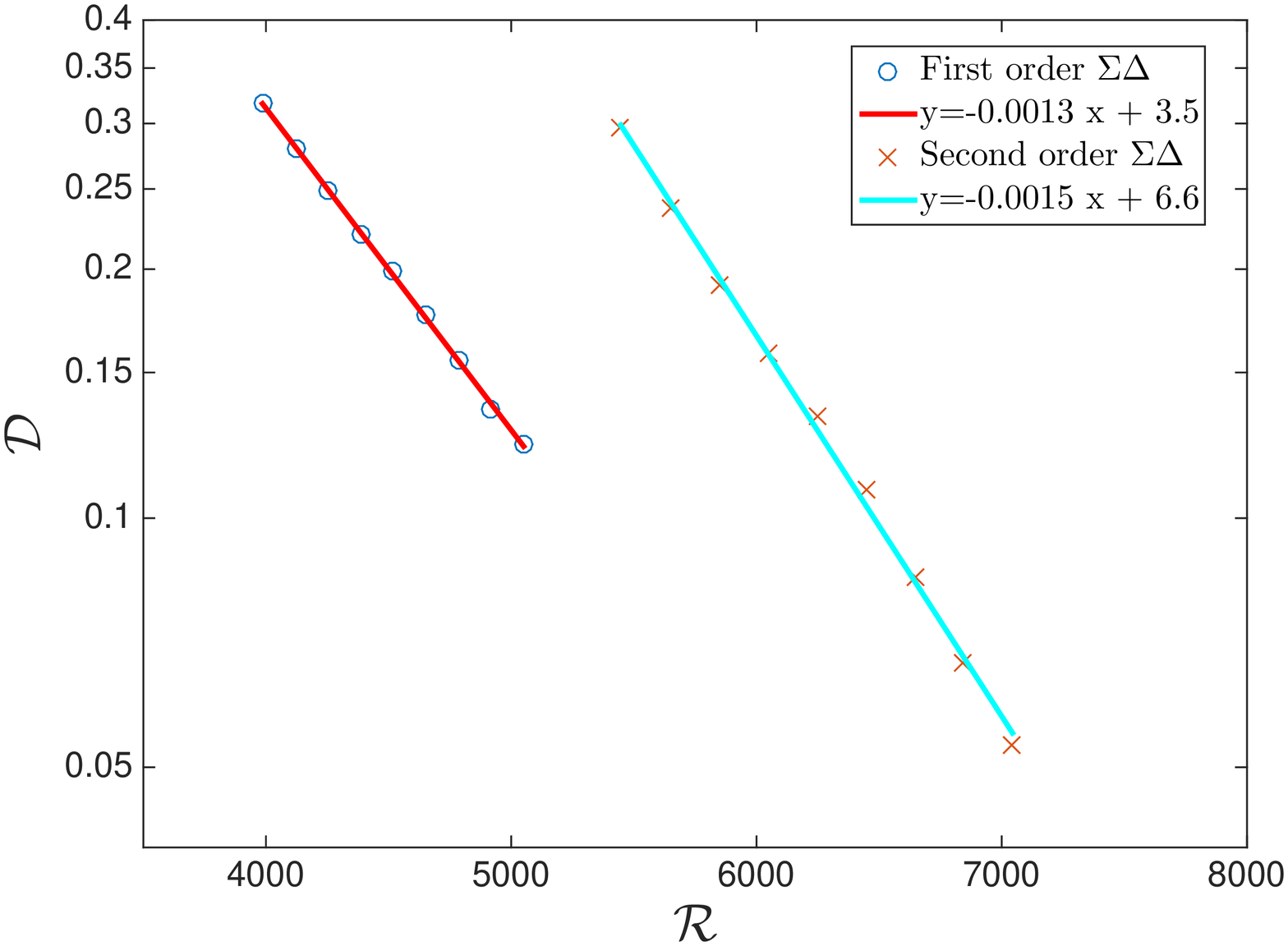}
\caption{Semi-log plot of $\mathcal{D}$ versus $\mathcal{R}$ where  $\mathcal{D}$ is the average error of reconstructions of 50 sparse signals.}
\label{Fig:sparse}
\end{center}
\end{figure}

\noindent\textbf{Experiment 2: Compressible signals}. This experiment illustrates the polynomial rate-distortion relationship for compressible signals (Corollary \ref{cor:compr}).  Here, we again compute the distortion $\mathcal{D}$ for different values of $\mathcal{R}$, where both change through varying $m\in [m_{\min}, m_{\max}]$ with $m_{\min}=\lfloor 10^{2.1}\rfloor $ and $m_{\max}=\lfloor 10^{3.2} \rfloor $. Similar to the previous experiment,  the matrix $B_m$ is formed by taking the upper left block of size $L\times m$ from a predefined big Bernoulli matrix $B$ of size $m_{\max} \times m_{\max}$. For each $m$, $L$ is calculated via Equation \eqref{eq:mL} with the $\Cr{RIPc}$ chosen heuristically but fixed for all choices of $m$. Here, we set $N=1200$, $\delta=0.1$, $K=20$, $T=500$, $r=2$, and we generate compressible signals from the $w\ell_p$ ball, with $p=1/3$ and $p=1/4$. For these choices of $p$, the results are reported in Figure \ref{Fig:compr}, where $\mathcal{D}$ is the average reconstruction error over 50 independent compressible signals. The number of bits $\mathcal{R}$ is computed via \eqref{eq:rate1}. Note that Figure \ref{Fig:compr} (which is the log-log plot of $\mathcal{D}$ with respect to $\mathcal{R}$)  shows that the fitting lines have slopes close to $1/p-1/2$ (Corollary \ref{cor:compr}) for both choices of $p$.
 \begin{figure}[htbp]
\centering
\includegraphics[width=9cm, height=7cm]{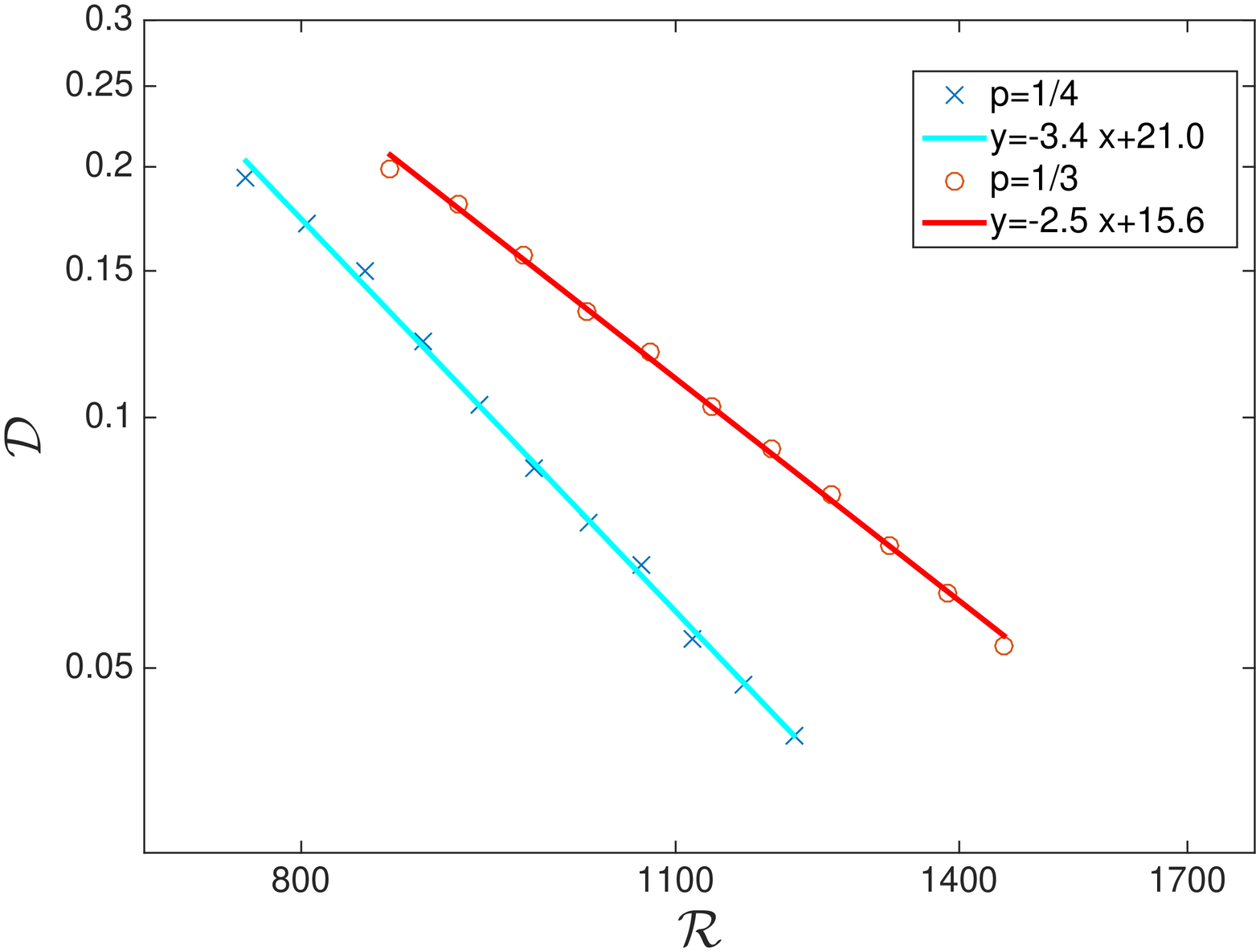}
\caption{Log-log plot of $\mathcal{D}$ versus $\mathcal{R}$, where  $\mathcal{D}$ is the average error of reconstructions of compressible signals from the second order $\sd$ quantization.}
\label{Fig:compr}
\end{figure}

 \begin{figure}[htbp]
\centering
\includegraphics[width=9cm, height=7cm]{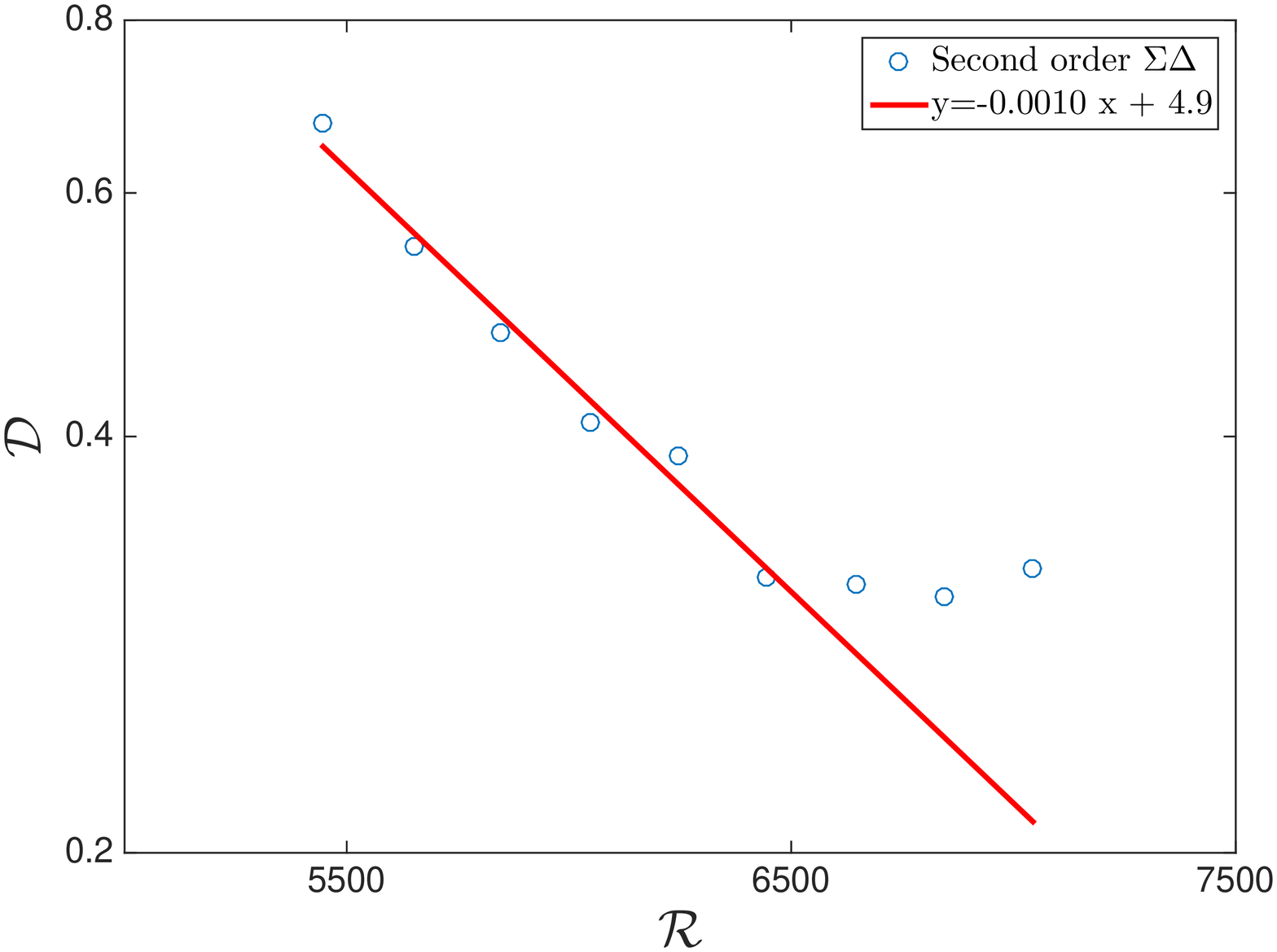}
\caption{Semi-log plot of $\mathcal{D}$ versus $\mathcal{R}$ where  $\mathcal{D}$ is the maximum error of reconstructions for 50 sparse signals via second order $\sd$ quantization and noisy measurements}
\label{Fig:noisy}
\end{figure}

\noindent\textbf{Experiment 3: Noisy measurements}. Using the same setting as in Experiment 1 but with noisy measurements $y=Ax+e$, we observe that the  error decay stops as soon as it reaches the level of the noise (see Figure \ref{Fig:noisy}). Here each entry of $e$ is i.i.d., drawn from the uniform distribution on $[-0.05,0.05]$. Each point in the figure is the maximum error over 100 independent trials. 
\section{Acknowledgements} 
The authors thank Sinan G{\"u}nt{\"u}rk for useful conversations and for sharing his proof of Lemma \ref{lm:sinan}. 
R. Saab has been supported in part by a UCSD  Research Committee Award, a Hellman Fellowship, and the NSF under grant DMS 1517204. R.~Wang was funded in part by an NSERC
Collaborative Research and Development Grant DNOISE II (22R07504). {\"O}.~Y{\i}lmaz was funded in part by a Natural Sciences and
Engineering Research Council of Canada (NSERC) Discovery Grant
(22R82411), an NSERC Accelerator Award (22R68054) and an NSERC
Collaborative Research and Development Grant DNOISE II (22R07504).


\bibliographystyle{IEEEtran}

\bibliography{DCCbib,nsfrefs,refs}

\begin{thebibliography}{10}
\providecommand{\url}[1]{#1}
\csname url@samestyle\endcsname
\providecommand{\newblock}{\relax}
\providecommand{\bibinfo}[2]{#2}
\providecommand{\BIBentrySTDinterwordspacing}{\spaceskip=0pt\relax}
\providecommand{\BIBentryALTinterwordstretchfactor}{4}
\providecommand{\BIBentryALTinterwordspacing}{\spaceskip=\fontdimen2\font plus
\BIBentryALTinterwordstretchfactor\fontdimen3\font minus
  \fontdimen4\font\relax}
\providecommand{\BIBforeignlanguage}[2]{{%
\expandafter\ifx\csname l@#1\endcsname\relax
\typeout{** WARNING: IEEEtran.bst: No hyphenation pattern has been}%
\typeout{** loaded for the language `#1'. Using the pattern for}%
\typeout{** the default language instead.}%
\else
\language=\csname l@#1\endcsname
\fi
#2}}
\providecommand{\BIBdecl}{\relax}
\BIBdecl

\bibitem{LaskaA2I}
J.~Laska, S.~Kirolos, M.~Duarte, T.~Ragheb, R.~Baraniuk, and Y.~Massoud,
  ``Theory and implementation of an analog-to-information converter using
  random demodulation,'' in \emph{Circuits and Systems, 2007. ISCAS 2007. IEEE
  International Symposium on}, May 2007, pp. 1959--1962.

\bibitem{LaskaA2I_2}
J.~Laska, S.~Kirolos, Y.~Massoud, R.~Baraniuk, A.~Gilbert, M.~Iwen, and
  M.~Strauss, ``Random sampling for analog-to-information conversion of
  wideband signals,'' in \emph{Design, Applications, Integration and Software,
  2006 IEEE Dallas/CAS Workshop on}, Oct 2006, pp. 119--122.

\bibitem{JLoriginal}
W.~Johnson and J.~Lindenstrauss, ``Extensions of {L}ipschitz mappings into a
  {H}ilbert space,'' \emph{Contemporary mathematics}, vol.~26, pp. 189--206,
  1984.

\bibitem{ailon2006approximate}
N.~Ailon and B.~Chazelle, ``Approximate nearest neighbors and the fast
  johnson-lindenstrauss transform,'' in \emph{Proceedings of the thirty-eighth
  annual ACM symposium on Theory of computing}.\hskip 1em plus 0.5em minus
  0.4em\relax ACM, 2006, pp. 557--563.

\bibitem{ailon2009fast}
------, ``The fast johnson-lindenstrauss transform and approximate nearest
  neighbors,'' \emph{SIAM Journal on Computing}, vol.~39, no.~1, pp. 302--322,
  2009.

\bibitem{krahmer2011new}
F.~Krahmer and R.~Ward, ``New and improved johnson-lindenstrauss embeddings via
  the restricted isometry property,'' \emph{SIAM Journal on Mathematical
  Analysis}, vol.~43, no.~3, pp. 1269--1281, 2011.

\bibitem{vybiral2011variant}
J.~Vyb{\'\i}ral, ``A variant of the {Johnson--Lindenstrauss} lemma for
  circulant matrices,'' \emph{Journal of Functional Analysis}, vol. 260, no.~4,
  pp. 1096--1105, 2011.

\bibitem{ve12-1}
R.~{V}ershynin, ``{I}ntroduction to the non-asymptotic analysis of random
  matrices,'' in \emph{{C}ompressed {S}ensing: {T}heory and {A}pplications},
  Y.~{E}ldar and G.~{K}utyniok, Eds.\hskip 1em plus 0.5em minus 0.4em\relax
  Cambridge: {C}ambridge {U}niv {P}ress, 2012, pp. xii+544.

\bibitem{SWY15}
\BIBentryALTinterwordspacing
R.~Saab, R.~Wang, and {\"{O}}.~Yilmaz, ``Quantization of compressive samples
  with stable and robust recovery,'' \emph{CoRR}, vol. abs/1504.00087, 2015.
  [Online]. Available: \url{http://arxiv.org/abs/1504.00087}
\BIBentrySTDinterwordspacing

\bibitem{GLPSY13}
C.~G{\"u}nt{\"u}rk, M.~Lammers, A.~Powell, R.~Saab, and {\"O}.~Y{\i}lmaz,
  ``Sobolev duals for random frames and sigma-delta quantization of compressed
  sensing measurements,'' \emph{Foundations of Computational mathematics},
  vol.~13, no.~1, pp. 1--36, 2013.

\bibitem{KSY13}
F.~Krahmer, R.~Saab, and {\"O}.~Y{\i}lmaz, ``{Sigma-Delta} quantization of
  sub-{G}aussian frame expansions and its application to compressed sensing,''
  \emph{Information and Inference}, vol.~3, no.~1, pp. 40--58, 2013.

\bibitem{CRT05}
E.~J. Cand{\`e}s, J.~Romberg, and T.~Tao, ``Stable signal recovery from
  incomplete and inaccurate measurements,'' \emph{Communications on Pure and
  Applied Mathematics}, vol.~59, pp. 1207--1223, 2006.

\bibitem{cohen2009compressed}
A.~Cohen, W.~Dahmen, and R.~DeVore, ``Compressed sensing and best $k$-term
  approximation,'' \emph{Journal of the American mathematical society},
  vol.~22, no.~1, pp. 211--231, 2009.

\bibitem{edmunds1996function}
D.~E. Edmunds and H.~Triebel, \emph{Function spaces, entropy numbers,
  differential operators}.\hskip 1em plus 0.5em minus 0.4em\relax Cambridge,
  UK: Cambridge Univ. Press, 1996.

\bibitem{kuhn2001lower}
T.~K{\"u}hn, ``A lower estimate for entropy numbers,'' \emph{Journal of
  Approximation Theory}, vol. 110, no.~1, pp. 120--124, 2001.

\bibitem{schutt1984entropy}
C.~Sch{\"u}tt, ``Entropy numbers of diagonal operators between symmetric
  {B}anach spaces,'' \emph{Journal of Approximation Theory}, vol.~40, no.~2,
  pp. 121--128, 1984.

\bibitem{daub-dev}
I.~Daubechies and R.~DeVore, ``{Approximating a bandlimited function using very
  coarsely quantized data: a family of stable sigma-delta modulators of
  arbitrary order},'' \emph{Annals of Mathematics}, vol. 158, no.~2, pp.
  679--710, 2003.

\bibitem{NST96}
S.~R. Norsworthy, R.~Schreier, and G.~C. Temes, Eds.,
  \emph{Delta-Sigma-Converters: Theory, Design and Simulation}.\hskip 1em plus
  0.5em minus 0.4em\relax Wiley-IEEE, 1996.

\bibitem{BB_DCC07}
\BIBentryALTinterwordspacing
P.~T. Boufounos and R.~G. Baraniuk, ``Quantization of sparse representations,''
  in \emph{Rice University ECE Department Technical Report 0701. Summary
  appears in Proc. Data Compression Conference (DCC)}, Snowbird, UT, March
  27-29 2007. [Online]. Available: \url{http://hdl.handle.net/1911/13034}
\BIBentrySTDinterwordspacing

\bibitem{BJKS15}
\BIBentryALTinterwordspacing
P.~T. Boufounos, L.~Jacques, F.~Krahmer, and R.~Saab, ``Quantization and
  compressive sensing,'' \emph{CoRR}, vol. abs/1405.1194, 2014. [Online].
  Available: \url{http://arxiv.org/abs/1405.1194}
\BIBentrySTDinterwordspacing

\bibitem{duarte2008single}
M.~F. Duarte, M.~A. Davenport, D.~Takhar, J.~N. Laska, T.~Sun, K.~E. Kelly,
  R.~G. Baraniuk \emph{et~al.}, ``Single-pixel imaging via compressive
  sampling,'' \emph{IEEE Signal Processing Magazine}, vol.~25, no.~2, p.~83,
  2008.

\bibitem{willett2011compressed}
R.~M. Willett, R.~F. Marcia, and J.~M. Nichols, ``Compressed sensing for
  practical optical imaging systems: a tutorial,'' \emph{Optical Engineering},
  vol.~50, no.~7, pp. 072\,601--072\,601, 2011.

\bibitem{wagadarikar2008single}
A.~Wagadarikar, R.~John, R.~Willett, and D.~Brady, ``Single disperser design
  for coded aperture snapshot spectral imaging,'' \emph{Applied optics},
  vol.~47, no.~10, pp. B44--B51, 2008.

\bibitem{baraniuk2014exponential}
R.~Baraniuk, S.~Foucart, D.~Needell, Y.~Plan, and M.~Wootters, ``Exponential
  decay of reconstruction error from binary measurements of sparse signals,''
  \emph{arXiv preprint arXiv:1407.8246}, 2014.

\bibitem{CGKSY15}
E.~Chou, C.~S. G{\"u}nt{\"u}rk, F.~Krahmer, R.~Saab, and {\"O}.~Y{\i}lmaz,
  ``Noise-shaping quantization methods for frame-based and compressive sampling
  systems,'' \emph{arXiv preprint arXiv:1502.05807}, 2015.

\bibitem{boufounos2014quantization}
P.~Boufounos, L.~Jacques, F.~Krahmer, and R.~Saab, ``Quantization and
  compressive sensing,'' 2014.

\bibitem{Donoho2006_CS}
D.~Donoho, ``Compressed sensing.'' \emph{IEEE Transactions on Information
  Theory}, vol.~52, no.~4, pp. 1289--1306, 2006.

\bibitem{GVT98}
V.~Goyal, M.~Vetterli, and N.~Thao, ``Quantized overcomplete expansions in
  $\mathbb{R}^{N}$: analysis, synthesis, and algorithms,'' \emph{IEEE
  Transactions on Information Theory}, vol.~44, no.~1, pp. 16--31, Jan 1998.

\bibitem{zymnis2010compressed}
A.~Zymnis, S.~Boyd, and E.~Candes, ``Compressed sensing with quantized
  measurements,'' \emph{IEEE Signal Proc. Lett.}, vol.~17, no.~2, pp. 149--152,
  2010.

\bibitem{QIHT}
L.~Jacques, K.~Degraux, and C.~D. Vleeschouwer, ``Quantized iterative hard
  thresholding: Bridging 1-bit and high-resolution quantized compressed
  sensing,'' in \emph{Proc. Intl. Conf. Sampling Theory and Applications
  (SampTA 2013), arXiv:1305.1786}, Bremen, Germany, 2013, pp. 105--108.

\bibitem{Jacques2010}
L.~Jacques, D.~K. Hammond, and M.~J. Fadili, ``Dequantizing compressed sensing:
  When oversampling and non-gaussian constraints combine,'' \emph{IEEE
  Transactions on Information Theory}, vol.~57, no.~1, pp. 559--571, Jan. 2011.

\bibitem{Jacques2013}
------, ``Stabilizing nonuniformly quantized compressed sensing with scalar
  companders,'' \emph{IEEE Transactions on Information Theory}, vol.~5, no.~12,
  pp. 7969 -- 7984, Jan. 2013.

\bibitem{G-exp}
C.~G{\"u}nt{\"u}rk, ``One-bit sigma-delta quantization with exponential
  accuracy,'' \emph{Communications on Pure and Applied Mathematics}, vol.~56,
  no.~11, pp. 1608--1630, 2003.

\bibitem{DGK10}
P.~{D}eift, C.~S. G{\"u}nt{\"u}rk, and F.~{K}rahmer, ``An optimal family of
  exponentially accurate one-bit sigma-delta quantization schemes,''
  \emph{Communications on Pure and Applied Mathematics}, vol.~64, no.~7, pp.
  883--919, 2011.

\bibitem{inose1963unity}
H.~Inose and Y.~Yasuda, ``{A unity bit coding method by negative feedback},''
  \emph{Proceedings of the IEEE}, vol.~51, no.~11, pp. 1524--1535, 1963.

\bibitem{benedetto2006sigma}
J.~Benedetto, A.~Powell, and {\"O}.~Y{\i}lmaz, ``{Sigma-delta ($\Sigma\Delta$)
  quantization and finite frames},'' \emph{IEEE Transactions on Information
  Theory}, vol.~52, no.~5, pp. 1990--2005, 2006.

\bibitem{BP07}
B.~G. Bodmann and V.~I. Paulsen, ``Frame paths and error bounds for
  {S}igma--{D}elta quantization,'' \emph{Applied and Computational Harmonic
  Analysis}, vol.~22, no.~2, pp. 176--197, 2007.

\bibitem{BLPY10}
J.~Blum, M.~Lammers, A.~M. Powell, and {\"O}.~Y{\i}lmaz, ``Sobolev duals in
  frame theory and {S}igma-{D}elta quantization,'' \emph{Journal of Fourier
  Analysis and Applications}, vol.~16, no.~3, pp. 365--381, 2010.

\bibitem{KSW12}
F.~Krahmer, R.~Saab, and R.~Ward, ``Root-exponential accuracy for coarse
  quantization of finite frame expansions,'' \emph{IEEE Transactions on
  Information Theory}, vol.~58, no.~2, pp. 1069 --1079, February 2012.

\bibitem{GLPSY}
C.~G{\"u}nt{\"u}rk, M.~Lammers, A.~Powell, R.~Saab, and {\"O}.~Y{\i}lmaz,
  ``Sobolev duals for random frames and sigma-delta quantization of compressed
  sensing measurements,'' \emph{ArXiv preprint ArXiv:1002.0182G}, 2010.

\bibitem{BB_CISS08}
P.~T. Boufounos and R.~G. Baraniuk, ``1-bit compressive sensing,'' in
  \emph{Proc. Conf. Inform. Science and Systems (CISS)}, Princeton, NJ, March
  19-21 2008.

\bibitem{jacques2013robust}
L.~Jacques, J.~N. Laska, P.~T. Boufounos, and R.~G. Baraniuk, ``Robust 1-bit
  compressive sensing via binary stable embeddings of sparse vectors,''
  \emph{IEEE Transactions on Information Theory}, vol.~59, no.~4, pp.
  2082--2102, 2013.

\bibitem{PV13}
Y.~Plan and R.~Vershynin, ``One-bit compressed sensing by linear programming,''
  \emph{Communications on Pure and Applied Mathematics}, vol.~66, no.~8, pp.
  1275--1297, 2013.

\bibitem{ai2014one}
A.~Ai, A.~Lapanowski, Y.~Plan, and R.~Vershynin, ``One-bit compressed sensing
  with non-gaussian measurements,'' \emph{Linear Algebra and its Applications},
  vol. 441, pp. 222--239, 2014.

\bibitem{knudson2014one}
K.~Knudson, R.~Saab, and R.~Ward, ``One-bit compressive sensing with norm
  estimation,'' \emph{arXiv preprint arXiv:1404.6853}, 2014.

\bibitem{chou2013beta}
E.~Chou, ``Beta-duals of frames and applications to problems in quantization,''
  Ph.D. dissertation, New York University, 2013.

\bibitem{daubechies2002beta}
I.~Daubechies, R.~DeVore, C.~G{\"u}nt{\"u}rk, and V.~Vaishampayan, ``Beta
  expansions: a new approach to digitally corrected a/d conversion,'' in
  \emph{Circuits and Systems, 2002. ISCAS 2002. IEEE International Symposium
  on}, vol.~2, 2002, pp. II--784.

\bibitem{daubechies2006robust}
I.~Daubechies and {\"O}.~Yilmaz, ``Robust and practical analog-to-digital
  conversion with exponential precision,'' \emph{IEEE Transactions on
  Information Theory}, vol.~52, no.~8, pp. 3533--3545, 2006.

\bibitem{IS13}
M.~Iwen and R.~Saab, ``Near-optimal encoding for sigma-delta quantization of
  finite frame expansions,'' \emph{Journal of Fourier Analysis and
  Applications}, pp. 1--19, 2013.

\bibitem{Foucart13}
S.~Foucart, ``Stability and robustness of $\ell_1$-minimizations with {Weibull}
  matrices and redundant dictionaries,'' \emph{Linear Algebra and its
  Applications}, vol. 441, pp. 4--21, 2014.

\bibitem{feldheim2010}
O.~N. Feldheim and S.~Sodin, ``A universality result for the smallest
  eigenvalues of certain sample covariance matrices,'' \emph{Geometric And
  Functional Analysis}, vol.~20, no.~1, pp. 88--123, 2010.

\bibitem{baraniuk2008simple}
R.~Baraniuk, M.~Davenport, R.~DeVore, and M.~Wakin, ``A simple proof of the
  restricted isometry property for random matrices,'' \emph{Constructive
  Approximation}, vol.~28, no.~3, pp. 253--263, 2008.

\bibitem{foucart2013mathematical}
S.~Foucart and H.~Rauhut, ``A mathematical introduction to compressive
  sensing,'' \emph{Basel: Birkh{\"a}user, Boston}, vol.~1, no.~3, 2013.

\end{thebibliography}
\section{Appendix}
\noindent\emph{\textbf{Proof of Lemma \ref{lm:sinan} (Due to Sinan G{\"u}nt{\"u}rk)}.}  Let $N_C := \# \{q: C_{E,Q}(q)\cap B_R\} $ be the number of quantization cells that intersect with $B_R$.  
Without loss of generality, we work with the infinite alphabet $\delta \mathbb{Z}$, with $\delta=1$. It is straightforward to generalize the argument to a finite alphabet and to show that the number of cells scales linearly with $1/\delta^k$. 

Let $y=Ex$ and fix the quantizer to be $Q_{\sd}^r$. Having $x\in C_{E,Q_{\sd}^r}(q)$ means that there exists $u\in \mathbb{R}^m$ with 
\begin{equation}\label{equ_cell}
y-q=D^r u,   \quad u_i \in [ -0.5, 0.5), \quad 1\leq i \leq m. 
\end{equation}
Let $\widetilde{y}= D^{-r}y$, $\widetilde{q}=D^{-r}q$ and $\widetilde{E}=D^{-r}E$. Plugging in \eqref{equ_cell} gives 
\[
\wtl y -\wtl q=  u, \quad u\in [ -0.5, 0.5), \quad 1\leq i \leq m
\]
which means $x \in C_{\wtl E, Q_0} (\wtl q)$, where $(Q_0(y))_i = \arg\min\limits_{n\in\Z} |n-y_i|$, for $i=1,...,m$. Since this  holds for all the $x$ in $C_{E,Q_r}(q)$, we then have
\[
C_{E,Q_r}(q)=C_{\widetilde{E},Q_0} (\widetilde{q}).
\] 
 Since $q$ and $\widetilde{q}$ are in one-to-one correspondence, it follows that 
$$
N_C = \# \{\wtl q:  \  C_{\widetilde{E},Q_0}(\wtl q) \cap B_R \}.
$$
Let $(\widetilde{e_i})_1^m$ be the rows of $\widetilde{E}$. The interior of $C_{\widetilde{E},Q_0}(\widetilde{q})$ is given by 
\[
\bigcap\limits_{i=1}^m \{x\in \mathbb{R}^k: \langle \wtl e_i, x \rangle \in (\wtl q_i -0.5,\wtl q_i+0.5)\}.
\]
For each $1\leq i\leq m$ and each $l\in \mathbb{Z}$, consider the hyperplanes
\[
H_{i,l}:=\{x\in \mathbb{R}^k: \langle \wtl e_i,x \rangle = (l-0.5) \}.
\]
and note that the collection of hyperplanes $\{H_{i,l}: 1\leq i\leq m, l \in \mathbb{Z}\}$ determines the set of all cells, i.e., $\{C_{\wtl E,Q_0} (\wtl q): \wtl q \in \mathcal{A}^m \}$. 
%
Now let us focus on the cells intersecting the ball $B_R$. For any $x\in B_R$, the Cauchy-Schwarz inequality implies
\[
|\langle \wtl e_i, x \rangle| \leq \|\wtl e_i\|_2 R.
\]
Hence 
\[
\# \{l: H_{i,l} \cap B_R \neq \emptyset \}  \leq 1+2R\|\wtl e_i\|_2,
\]
and
\[
n_R := \# \{(i,l): H_{i,l} \cap B_R \neq \emptyset\} \leq m+2R \sum\limits_{i=1}^m \|\wtl e_i\|_2.
\]
What remains is to bound the sum above. To that end, note that
\[
\sum\limits_{i=1}^m \|\widetilde{e}_i\|_2 \leq \sqrt m \|\wtl E\|_F \leq m^{r+1/2} \|E\|_F
\]
since $\|D^{-r} \|_2 \leq m^r $. Now, observe that $\|E\|_F^2$ is simply the  sum of the squares of independent sub-Gaussian variables, hence a sum of independent sub-exponential variables. By the large deviation result on such sub-exponential variables (see e.g., Section 5.2.4 \cite{ve12-1}),  for each $\alpha >0$ with probability exceeding $1-e^{-\alpha^2 mk}$, we have $\|E\|_F \leq (1+\alpha) \sqrt {mk}$. Thus we have 
\[
\sum\limits_{i=1}^m \|\wtl e_i\|_2 \leq (1+\alpha)m^{r+1} \sqrt k, 
\]
which in turn gives
\[
n_R \leq m+2(1+\alpha)R m^{r+1} \sqrt k
\]
as an upper bound on the number of hyperplanes. On the other hand, it is well known that the total number of cells determined by $n$ hyperplanes in $\mathbb{R}^k$ is at most 
$$\sum_{i=0}^k \binom{n}{i}.$$
Set $n=n_R$, and note that  depending on the size and location of the ball $B_R$ there are two possibilities: \\
\textbf{Case 1}. $n_R >2k$. In this case, Stirling's approximation gives
\[
N_C \leq \sum\limits_{i=0}^k \binom {n_R}{i} \leq (k+1)\binom{n_R}{k} \leq (k+1)^{1/2} \left(\frac{me}{k}\right)^k (1+2(1+\alpha)Rm^r\sqrt k )^k.
\]
\textbf{Case 2}. $n_R \leq 2k$. In this case, we use the trivial bound $N_C \leq 2^{n_R}\leq 2^{2k}$.
Combining the two estimates 
\[
N_C \lesssim  \max\left\{2^{2k}, (1+\alpha)^k R^k \frac{m^{(r+1)k}}{k^{k/2-1/2}} \right\}.
\]
\qed
 
\end{document}